\documentclass[11pt]{article}

\usepackage{float}
\usepackage{caption}

\usepackage{tikz}

\usepackage{amsmath,amsthm,amsfonts,amssymb,amsopn,amscd} 
\usepackage{color}

\def\tilde{\widetilde}

\def\Ad{{\hbox{\rm Ad}}}

\def\a{\alpha}
\def\b{\beta}

\def\e{\varepsilon}
\def\g{\gamma}

\def\l{\lambda}

\def\r{\rho}

\def\th{\theta}

\def\A{{\cal A}}

\def\M{{\cal M}}
\def\N{{\cal N}}
\def\R{{\cal R}}

\def\H{{\cal H}}
\def\K{{\cal K}}
\def\S{{\cal S}}

\def\f{{\varphi}}
\def\s{{\sigma}}
\def\p{{\pi}}
\def\l{{\lambda}}
\def\x{{\xi}}

\def\PSL{{{\rm PSL}(2,\mathbb R)}}

\def\S2{S^{1(2)}}
\def\Poi{{\cal P}_+^\uparrow}

\def\uPoi{\overline{\cal P}_+^\uparrow}

\def\Reali{\mathbb R}
\newtheorem{theorem}{Theorem}[section]
\newtheorem{lemma}[theorem]{Lemma}

\newtheorem{corollary}[theorem]{Corollary}

\newtheorem{proposition}[theorem]{Proposition}

\theoremstyle{definition} 

\theoremstyle{remark} \newtheorem{remark}[theorem]{Remark}

\newcommand{\ben}{\begin{equation}}
\newcommand{\een}{\end{equation}}

\def\PSL{PSU(1,1)}

\def\SL2{{{\rm SL}(2,\R)}}

\def\PSL2{{{\rm PSL}(2,\Reali)}}

\def\U1{{{\rm V}(1)}}
\def\SU2{{{\rm SV}(2)}}

\def\SU{{{\rm SU}}}

\def\A{{\mathcal A}}

\def\H{{\mathcal H}}

\def\K{{\mathcal K}}

\def\M{{\mathcal M}}
\def\N{{\mathcal N}}

\title{\Huge{Entropy distribution of localised states }}

\author{{\sc Roberto Longo\thanks{Supported by the ERC Advanced Grant 669240 QUEST ``Quantum Algebraic Structures and Models'', MIUR FARE R16X5RB55W  QUEST-NET and GNAMPA-INdAM.}
}
\\
Dipartimento di Matematica,
Universit\`a di Roma ``Tor Vergata'',\\
Via della Ricerca Scientifica, 1, I-00133 Roma, Italy\\
E-mail: {\tt longo@mat.uniroma2.it}
}

\date{}
\begin{document}

\maketitle

\begin{abstract}
We study the geometric distribution of the relative entropy of a charged localised state in Quantum Field Theory. With respect to translations, the second derivative of the vacuum relative entropy is zero out of the charge localisation support and positive in mean over the support of any single charge. For a spatial strip, the asymptotic mean entropy density is $\pi  E$, with $E$ the corresponding vacuum charge energy. In a conformal QFT, for a charge in a ball of radius $r$, the relative entropy is non linear, the asymptotic mean radial entropy density is $\pi  E$ and Bekenstein's bound is satisfied. We also study the null deformation case. We construct, operator algebraically, a positive selfadjoint operator that may be interpreted as the deformation generator, we thus get a rigorous form of the Averaged Null Energy Condition that holds in full generality. In the one dimensional conformal $U(1)$-current model, we give a complete and explicit description of the entropy distribution of a localised charged state in all points of the real line; in particular, the second derivative of the relative entropy is strictly positive in all points where the charge density is non zero, thus the Quantum Null Energy Condition holds here for these states and is not saturated in these points. 
\end{abstract}

\newpage

\section{Introduction}
The main aim of this paper is to provide a rigorous, model independent, Operator Algebraic approach to the Averaged Null Energy Condition (ANEC) and Quantum Null Energy Condition (QNEC) in Quantum Field Theory. We begin by recalling a few basic facts. 
\medskip

\noindent
{\it Entropy in QFT.} Quantum Information is having an increasingly important interplay with Quantum Field Theory, that naturally took place in the framework of quantum black hole thermodynamics. The first non commutative entropy notion, von Neumann's quantum entropy, was originally designed as a Quantum Mechanics version of Shannon's entropy: if a state $\psi$ is given a density matrix $\r_\psi$
\[
\text{von Neumann entropy of $\psi$}: \ - {\rm Tr}(\r_\psi\log\r_\psi) \ .
\]
However, in Quantum Field Theory, local von Neumann algebras are typically factors of type $III_1$ (see \cite{L82}), no trace or density matrix exists and von Neumann entropy is undefined. Nonetheless, the Tomita-Takesaki modular theory applies and one 
may considers the relative entropy \cite{Ar}
 \[
\text{Araki relative entropy}:\  S(\f |\!|\psi) = -(\xi, \log\Delta_{\eta,\xi}\xi)
\]
between $\psi$ and the vacuum state $\f$, that extends Umegaki's type $I$ notion
\[
S(\f |\!|\psi) = {\rm Tr}\big(\r_\f(\log\r_\f - \log \r_\psi)\big)
\]
and is defined in general; here $\xi$ and $\eta$ are cyclic vector representatives of $\f$, $\psi$ and  $\Delta_{\eta,\xi}$ is the associated relative modular operator (see \cite{OP,T} and Sect. \ref{RE}). $S(\f |\!|\psi)$ subtracts indeed the entanglement ultraviolet divergences, common both to $\f$ and $\psi$, and is finite for a dense set of states. 

$S(\f |\!|\psi)$ generalises the classical Kullback-Leibler divergence and measures how $\psi$ deviates from $\f$. 
From the information theoretical viewpoint, $S(\f |\!|\psi)$ is the mean value in the state $\f$ of the difference
between the information carried by the state $\psi$ and the state $\f$. 

In \cite{L97}, we made a relative entropy analysis in black hole thermodynamics. We consider relative entropy as a primary concept and other entropy quantities as derived concepts; this point of view is also implicit or explicit in several recent papers (e.g. \cite{BlCa}). A first rigorous computation of the relative entropy (mutual information) in QFT has been done in \cite{LXrelCFT} in chiral conformal QFT for subnets of Fermions and their finite index extensions, that cover all known unitary cases.

Operator Algebras provide the natural framework for quantum entropy issues, see \cite{L18, LXbek, Wit, Ho18}. 
In this paper we shall pursue a rigorous, general analysis of the entropy distribution of localised states in Quantum Field Theory. 
\medskip

\noindent
{\it Energy bounds in QFT.}
Energy conditions play classically an important role in general relativity, see \cite{HE, Wit81}.
In Quantum Field Theory there appear non-local conditions such as the ANEC (see \cite{FR,V}), 
namely the integrated energy density along an entire null ray is to be positive; yet 
the energy may locally have negative density states \cite{EGJ}, although energy lower bounds hold true even locally. In particular, in conformal QFT, model independent local lower bounds have been obtained in \cite{FH,Wien06}. 

Motivated by the study of relativistic quantum field theory coupled to gravity, Bousso, Fisher, Liechenauer, and Wall \cite{BFKLW} proposed the QNEC to the effect that the null energy density at a point is bounded below by the second derivative of a suitable quantum entropy. For a deformation $u$-null direction, the relation may be formally written (in natural units) as
\[
\langle T_{uu}\rangle \geq \frac{1}{2\pi} S''_A(t) \ .
\]
Here $T$ is the stress-energy tensor,  $S_A$ is the entropy relative to the region $A$ on one side of the deformation and $S_A''$ is the second derivative of $S_A$ with respect to the deformation parameter $t$.  

In this context, the positivity of the second derivative of the relative entropy
\[
S''(t)\geq 0
\]
appears unexpectedly and is supposed to hold in general in Quantum Field Theory.\footnote{The convexity of $S$ is not an intrinsic concept, as it depends on re-parameterizing $t$. However, in Section \ref{null}, we shall have a natural deformation parameter $t$, the half-sided modular translation parameter.}    

QNEC is expected to hold model independently, see \cite{BFKW}. It relates two conceptually different quantities  and
deep effort  in various directions has been done so far, e.g. in holographic theories (see \cite{Ha} for a review), with the aim to understand it. 

With focus on the relative entropy, here we are going to consider a state $\psi$ that arises as local charge excitation of the vacuum state $\f$ and consider  the relative entropy $S(\f |\!|\psi)$, where $\f$ and $\psi$ are restricted to the local von Neumann algebra $\A(O)$ of a spacetime region $O$; we then study the dependence of $S(\f |\!|\psi)$ on $O$. We shall also study the similar dependence of $S(\psi |\!|\f)$ on $O$, that is more directly connected with the ANEC. 

We now summarise our main results. 
\medskip

\noindent
{\it Entropy density, ANEC and QNEC.} We consider DHR charges (charges with short range interaction \cite{DHR, DHR2}) or, more generally,  charges localisable in  unbounded spacelike cones (charges of electromagnetic type \cite{BF}), see Section \ref{charges}. If $\psi$ is the state obtained from the vacuum $\f$ by adding a charge $\r$ localised inside a (Rindler) wedge region $W$, then the following formula holds \cite{L97}:
\ben\label{KWS}
S(\f|\!|\psi)  = 2\pi E_{\rm loc} + S(\r)  \ ,
\een
where $\f$ and $\psi$ are restricted to $\A(W)$, $E_{\rm loc} = (\xi, K_{\r,W}\xi)$ is the charge local energy, with $K_{\r,W}$ the Rindler modular Hamiltonian in presence of the charge $\r$, and 
\[
S(\r) = \log d(\r)
\] 
is half of the conditional entropy of $\r$, independent of $\psi$ (see Sect. \ref{charges}); $d(\r)$ is the DHR statistical dimension \cite{DHR}, which is equal the square root of the Jones index of $\r$ \cite{L89,L90}. We 
shall also have a corresponding formula for $S(\psi|\!|\f)$. 

Formula \eqref{KWS} is the key to study the dependence of $S(\f |\!|\psi) $ on $W$. In particular, if $W_t$ is the shifted wedge by a null translation by $t \geq 0$, $\r_1,\dots \r_n$ are charges localised in $W$
and $S(t)$ is the relative entropy between $\f|_{\A(W_t)}$ and $\psi|_{\A(W_t)}$ we have
\[
S''(t) = 0 \ ,
\]
with $t$ in any interval $(a,b)$ such that $W_b$ contains the support of some of the charges $\r_k$ and the support of the remaining charges is contained in $W'_a$, while
\[
\int^b_a S''(t) {\rm d}t >0 \ ,
\]
if $W_a\cap W'_b$ contains the support of at least one of the $\r_k$'s and the remaining charges' support is contained either in $W_b$ or $W'_a$.

In the conformal case, we study the relative entropy $S(r) =  S(\f_r |\!| \psi_r)$ relative to a double cone $O_r$ of radius $r$ with center at the origin. If the charges are localised in $O_{r_0}$, then
\[
S(r) = \pi\Big({r E - \frac1r E'}\Big) + \log d(\r)\ , \quad r \geq r_0\ ,
\]
with $E, E'$ the mean vacuum energies corresponding to the modular Hamiltonians of $\r$ w.r.t. $O_1$ and $O'_1$. The corresponding inequality for $ S(\psi_r |\!| \f_r)$ implies Bekenstein's bound \cite{Be} (cf. \cite{BlCa,LXbek}). 

We shall then make an analysis for general null deformations of the wedge $W$, that is of main importance in view of the ANEC and QNEC. 
Concerning the ANEC, we first abstractly construct a positive selfadjoint operator that we then interpret as the deformation generator, that so becomes manifestly positive.
To this end, we shall need to extend a structure theorem for half-sided modular inclusions of von Neumann algebras \cite{Wi1, AZ}.  

Let $f$ be a function on the null horizon of $W$ giving the boundary deformation and $\r$ a charge localised on the null horizon. Our positive operator $H_{\r,f}$ is the translation generator associated with the half-sided modular inclusion of local von Neumann algebras $\A(W_f)\subset\A(W)$ relative to both the vacuum state $\f$ and the charged state $\psi$.  It satisfies
\[
E_{\r,f} = \frac{1}{2\pi} \int_{-\infty}^{+\infty} S''(t){\rm d}t  > 0 \ ,
\]
where $E_{\r,f} = (\xi, H_{\r,f}\xi)$ is the vacuum energy associated with $\r$ and the null deformation given by $f$, and $t$ is the deformation parameter.  

Now, based on a physical argument in terms of the stress-energy tensor $T_{\mu\nu}$ associated with $\r$  \cite{LLS}, we may interpret the vacuum half-sided modular Hamiltonian $H_{f}$ as
\[
H_{f}  = 2\pi\int f(y)T_{uu}{\rm d}^{n-1}y {\rm d}u  \ ,
\]
(null coordinates), so $H_{f} \geq 0$ is a form of the ANEC that holds true in full generality. 

Our analysis is explicit and complete in the one-dimensional model given by the $U(1)$-current. Here the charges are associated with real functions $\ell$ in the Schwartz space $S(\mathbb R)$ \cite{BMT}. We have
\[
S(t)  = \pi\!\int_{t}^{+\infty} (x-t)\ell^2(x){\rm d}x\ 
\]
at any point $t\in\mathbb R$. So
\[
S''(t) = \pi \ell^2(t) \ ,
\]
and we will infer that
\[
E(t) = \text{\small $\frac{1}{2\pi}$} S''(t)\ ,
\]
with $E(t)$ the energy density of the $\ell$-charge at $t$; thus here the inequality $E(t)\geq S''(t)/2\pi$ is an equality and the QNEC holds for the considered charged states and is not saturated in all points of positive energy density. 
 
\section{Mathematical background}
We begin by describing the mathematical set up that we later apply in the physical context. As a general reference for the theory of Operator Algebras, in particular for the Tomita-Takesaki modular theory, see \cite{T}. 
\subsection{Translations and the modular Hamiltonian}\label{transmod}
Let $\M$ be a von Neumann algebra on a Hilbert space $\H$ with cyclic and separating unit vector $\xi\in\H$. We assume there exists a one-parameter unitary group $T$ on $\H$, with positive/negative generator $H$, such that 
\ben\label{bcond}
T(t)\xi = \xi\quad {\rm and}\quad T(t)\M T(-t)\subset \M, \quad  t\geq 0 \ .
\een
Borchers' theorem \cite {Bo} gives, in particular, the commutation relations
\ben\label{borch}
\Delta_\xi^{is}T(t)\Delta_\xi^{-is} = T(e^{\mp 2\pi s}t) \ ,
\een
where $\Delta_\xi = \Delta_{\xi,\M}$ is the modular operator of $\M$ w.r.t. $\xi$, see \cite{T}. 

So we have a positive energy unitary representation of the ``$ax + b$'' group where the dilation one-parameter unitary group $D$ is given by $D(2\pi s) = \Delta_\xi^{\mp is}$, namely the selfadjoint generator $K$ of $D$ is
\[
2\pi K = \mp \log\Delta_\xi\ .
\]
By the commutation relation \eqref{cr} we then have 
\ben\label{shift}
T(t)\log\Delta_\xi T(-t) = \log\Delta_\xi \pm 2\pi t  H\ ,
\een
where $H$ is the selfadjoint generator of $T$. 

Setting now
\ben\label{trtun}
\M_t \equiv T(t)\M T(-t)
\een
we have a tunnel of von Neumann algebras
\[
\M_0\equiv\M \supset \M_t \supset \M_{t'}\ ,\quad t' > t > 0\ .
\]
The translation tunnel  \eqref{trtun} is also given by
\ben\label{tt}
\M_t = \s^\f_{s}(\N), \quad t =   e^{-2\pi s} - 1, \quad \mp s \geq 0\ .
\een
In particular, the inclusion $\M_1\subset \M$ is $\pm$hsm, see below Sect. \ref{Wies}.
\subsection{Relative entropy}\label{RE}
Let again $\M$ be a von Neumann algebra on a Hilbert space $\H$ and let $\f, \psi$ be faithful, normal positive linear functionals on $\M$. We may assume that $\M$ acts standardly on $\H$ so that there are cyclic and separating vectors  
$\xi,\eta\in\H$ giving the states $\f$ and $\psi$.  

The relative modular operator $\Delta_{\eta,\xi} = S^*_{\eta,\xi}S_{\eta,\xi}$ is given by the polar decomposition
\[
S_{\eta,\xi} = J \Delta^{1/2}_{\eta,\xi}
\]
with $S_{\eta,\xi}$ the closure of the antilinear operator $X\xi\mapsto X^*\eta$, $X\in\M$, on $\H$. 
Once we fix $\xi$, $\Delta_{\eta,\xi} \equiv \Delta_{\eta,\xi,\M}$ does not depend on the chosen vector representative $\eta$ of $\psi$ but only on $\psi$. 

The Connes Radon-Nikodym unitary cocycle \cite{C73} can be written as 
\[
w_s \equiv (D\psi : D\f)_s = \Delta_{\eta,\xi}^{is}\Delta_\xi^{-is},\quad s\in\mathbb R \ .
\]
We have $w_s \in \M$, $s\in \mathbb R$, and
\ben\label{ci}
\s^\psi_s(x) = w_s\s^\f_s(x)w_s^*,\quad x\in \M \ .
\een
If $\M$ is a factor, $w_s\in \M$ is uniquely determined by \eqref{ci}, the cocycle equation $w_{s+s'} = w_s\s_s^\f(w_{s'})$ and
\[
{\underset{s\, \longrightarrow\,  - i}{\rm anal.\, cont.\,}} \f(w_s) = \psi(1) \ .
\]
The relative entropy is given by\footnote{The notation $S(\psi |\!| \f) = - (\xi , \log\Delta_{\eta,\xi}\xi)$ was used in ref. \cite{Ar}; our notation is commonly used.}
\[
S(\f |\!| \psi) = - (\xi , \log\Delta_{\eta,\xi}\xi)\ ,
\]
and is well-defined by the spectral theorem \cite{Ar}. 
It depends on $\f$ and $\psi$ but not on the chosen vector representatives $\xi$, $\eta\in\H$. 
It satisfies, in particular, positivity and monotonicity:
\[
S(\f |\!| \psi) \geq 0\qquad {\rm and}\qquad S(\f|_\N |\!| \psi|_\N)  \leq S(\f |\!| \psi) \ ,
\]
with $\N\subset\M$ a von Neumann subalgebra (see \cite{OP,Wit}). 
We have
\[
S(\f |\!| \psi) = i\frac{\rm d}{{\rm d}s}\f(w_s)|_{s=0} 
\]
if $\f(w_s)$ is differentiable at zero, and $+\infty$ otherwise. As $(D\f : D\psi)_s =w^*_s$, we also have
\ben\label{w*}
S(\psi |\!| \f) = -i\frac{\rm d}{{\rm d}s}\psi(w_s)|_{s=0}  \ .
\een
If $\l,\mu>0$ we have
\ben\label{scalentr}
S(\l\f |\!| \mu\psi) = \l S(\f |\!| \psi)  - \l\log(\mu/\l)\f(1)\ .
\een
\subsection{Extending a theorem of Wiesbrock, Araki and Zsido}\label{Wies}
The results in this section will be used in Section \ref{null}, yet they clarify the underlying structure in all Section \ref{egd}. 

Let $\N\subset\M$ be an inclusion of von Neumann algebras on a Hilbert space $\H$ and $\f$ a faithful normal state of $\M$; we assume that $\f$ is given by a unit vector $\xi\in\H$ which is cyclic and separating for both $\N$ and $\M$. 

We shall say that the inclusion $\N\subset\M$ is $\pm$hsm (half-sided modular) with respect to $\f$ if
\[
\s^\f_{s}(\N)\subset \N,\quad \pm s\geq 0\ .
\]
The result in \cite{Wi1, AZ} gives a converse to Borchers' theorem: if $\N\subset\M$ is $\mp$hsm w.r.t. $\xi$ as above, then
\ben\label{WAZ}
H = \frac{1}{2\pi}(\log \Delta_{\xi,\N} - \log \Delta_{\xi,\M})
\een
is a positive, essentially selfadjoint operator. Denoting by the same symbol $H$ its closure, the unitary one parameter group $T(t) = e^{itH}$ satisfies \eqref{bcond} and $\N= T(\pm 1)\M T(\mp1)$. 

As a consequence, $\N\subset\M$ is -hsm w.r.t. $\xi$ iff $\M'\subset\N'$ is +hsm w.r.t. $\xi$. For simplicity, in the following we consider only -hsm inclusions, yet every statement will have a dual statement for +hsm inclusions.  

Let $\N\subset\M$ be a -hsm inclusion of von Neumann algebra w.r.t $\f$ as above. We then have the translation tunnel $\M_t$; 
note that $\M_t \subset \M$ is -hsm with respect to $\f$ if $t \geq 0$.  Note also that
\ben\label{vt}
\bigvee_{s\in\mathbb R} \s^\f_{s}(\N) = \bigvee_{t>0}\M_t = \M \ ,
\een
where the lattice symbol $\vee$ denotes the von Neumann algebra generated. 

Let $\psi$ be another faithful normal state of $\M$, given by the cyclic and separating vector $\eta\in\H$ for $\M$, and assume that the Connes Radon-Nikodym unitary cocycle is localised as follows:
\ben\label{ws}
w_s \equiv (D\psi : D\f)_s \in \M_R, \quad  s\leq 0\ ,
\een
for some $R\geq 1$. We then have:
\begin{lemma} 
$\N\subset\M$ is -hsm with respect to $\psi$. 
\end{lemma}
\begin{proof} Immediate because 
$\s^\psi_s(\N) = w_s\s^\f_s(\N)w^*_s \subset  w_s \N w^*_s = \N$, $  s\leq 0$.
\end{proof}
\begin{lemma}\label{nt}
$\s^\psi_{s}(\N) = \M_{t}$, $0 <t \leq R$, $ t =   e^{-2\pi s} - 1$.  
\end{lemma}
\begin{proof}
We have $\s_{s}^\psi(\N) = w_s \s_{s}^\f(\N)w^*_s$. On the other hand $1 = w_{s -s} = w_{s}\s^\f_{{s}}(w_{-s})$, namely
\ben\label{u*}
w_{s} = \s^\f_{s}(w^*_{-s}) \ .
\een
Let $s >0$. Since $w_{-s}\in\N$, it follows that $w_{s}\in \s^\f_{s}(\N) = \M_{t}$. Thus
\[
\s^\psi_{s}(\N) = w_{s} \s^\f_{s}(\N) w^*_{s} = w_{s} \M_{t} w^*_{s} = \M_{t}\ .
\]
Let now $1 \leq t \leq R$. Then again
$\s^\psi_{s}(\N) = w_{s} \s^\f_{s}(\N) w^*_{s} = w_{s} \M_{t} w^*_{s} = \M_{t}$
because $u_s\in\M_R \subset \M_t$. 
\end{proof}
\begin{proposition}
$\eta$ is cyclic for $\N$. 
\end{proposition}
\begin{proof}
Let $\eta'$ be a vector orthogonal to $\N\eta$. Then $h(s)\equiv (\eta', \Delta^{is}x\eta) =0$, $s \geq 0$, where $\Delta$ is the modular operator associated with $\M$, $\eta$ and $x$ belongs to $\N$. Since $x\eta$ belongs to the domain of $\Delta^{1/2}$, $h$ is the boundary value of a function analytic in the strip $-\frac12 < \Im z < 0$, thus $h$ vanishes identically and $\eta'$ is orthogonal to $(\vee_{t>0}\M_t)\eta= \M\eta$ by \eqref{vt}. So $\eta' =0$. 
\end{proof}
Consider now the $2\times 2$ matrix algebras over $\N$ and $\M$
\[
\tilde\N \equiv \N\otimes{\rm Mat}_2(\mathbb C),\qquad \tilde\M \equiv \M\otimes{\rm Mat}_2(\mathbb C)
\]
and $\vartheta$ the positive linear functional on $\tilde\M$ given by
\[
\vartheta\left(\begin{matrix}x_{11}& x_{12}\\
x_{21}& x_{22}\end{matrix}\right)
= \f(x_{11}) + \psi(x_{22}), \quad x_{ij}\in\M\ .
\]
\begin{proposition}\label{theta}
$\tilde\N\subset\tilde\M$ is -hsm with respect to $\vartheta$. 
Moreover $\s^\vartheta_s(\tilde\N) = \M_t\otimes{\rm Mat}_2(\mathbb C)$,  $ t \leq R$. 
\end{proposition}
\begin{proof}
We have \cite{C73}:
\ben\label{mat}
\s^\vartheta_s \left(\begin{matrix}x_{11}& x_{12}\\
x_{21}& x_{22}\end{matrix}\right)
=\left(\begin{matrix}\s^\f_s(x_{11})& \s_s^\f(x_{12})w^*_s\\
w_s \s_s^\f(x_{21})& \s^\psi_s(x_{22})\end{matrix}\right)
\een
thus, if $s \leq 0$ and all the $x_{ij}$ belong to $\N$, then all entries of the matrix $\s^\vartheta_s \left(\begin{matrix}x_{11}& x_{12}\\
x_{21}& x_{22}\end{matrix}\right)$ belong to $\N$ as well,  because $w_s\in\N$ and of the -hsm modularity assumption.   

The second statement also follows similarly as $w_s \in \M_R$ by Lemma \ref{nt}.  
\end{proof}
\begin{theorem}
\label{rw}
Let $\N\subset \M$ be -hsm w.r.t. $\xi, \eta$ with the property \eqref{ws} as above. Then 
\[
H \equiv \frac{1}{2\pi  }\big(\log\Delta_{\eta,\xi, \N}- \log\Delta_{\eta,\xi}\big)
\]
is a positive essentially selfadjoint operator and we denote also its closure by $H$. 

The one parameter unitary group $T$ generated by $H$  satisfies 
\[
T(t)\M T(-t)=\M_t,\qquad 
T(t)\log\Delta_{\eta,\xi}T(-t) =  \log\Delta_{\eta,\xi, \M_t}, \quad  0\leq t \leq R\ ,
\]
with $\M_t$ given by \eqref{tt}, and $T(t)$, $D(s)$, $t,s\in \mathbb R$, provide a representation of the ``$ax + b$'' group, with $D(2\pi s)\equiv \Delta^{-is}_{\eta,\xi}$, namely
\[
\Delta_{\eta,\xi}^{is}T(t)\Delta_{\eta,\xi}^{-is} = T(e^{-2\pi s}t),\quad s,t\in\mathbb R \ .
\]
\end{theorem}
\begin{proof}
The GNS Hilbert space $\tilde\H$ of $\vartheta$ is the direct sum of four copies of $\H$
\[
\tilde\H = \bigoplus_{i,j}\H_{ij}, \quad \H_{ij} = \H,\quad i,j = 1,2\ ,
\]
$\vartheta$ is given by the vector $\th = \xi\oplus\eta$ on $\H_{11}\oplus\H_{22}$ and the modular operator $\Delta_\th \equiv \Delta_{\th,\tilde\M}$ decomposes as
\[
\Delta_\th = {\sum_{i,j}}^\oplus\Delta_{ij}
\]
with
\[
\Delta_{11} = \Delta_{\xi,\M}, \ \ \Delta_{22} = \Delta_{\eta,\M}, \ \ \Delta_{12} = \Delta_{\xi,\eta,\M}, \ \ \Delta_{21} = \Delta_{\eta,\xi,\M}.
\]
We can make the corresponding decomposition for $\vartheta|_{\tilde\N}$ and $\Delta_{\th,\tilde\N}$
\[
\Delta_{11,\tilde\N} = \Delta_{\xi,\N}, \ \ \Delta_{22,\tilde\N} = \Delta_{\eta,\N}, \ \ \Delta_{12,\tilde\N} = \Delta_{\xi,\eta,\N}, \ \ \Delta_{21,\tilde\N} = \Delta_{\eta,\xi,\N}.
\]
Now, by the -hsm property in Prop. \ref{theta}, it follows by \eqref{WAZ} that
\ben\label{Ht}
\tilde H = \frac{1}{2\pi}(\log\Delta_{\th, \tilde\N} - \log\Delta_\th)
\een
is a positive, essentially selfadjoint operator on $\tilde\H$ and we have
\[
\tilde T(1 - e^{2\pi s}) = \Delta_{\th,\tilde\N}^{is} \Delta^{-is}_\th,
\]
with $\tilde T$ the one-parameter unitary group generated by the closure of $\tilde H$, that satisfies
\[
\tilde\M_t \equiv\tilde T(t) \tilde\M \tilde T(-t) \subset \tilde\M, \quad t\geq 0,\qquad \tilde T(1) \tilde\M \tilde T(-1) = \tilde\N\ .
\]
Clearly $\tilde T(t)\theta = \theta$ and we have the decomposition
\[
\tilde T(t) = {\sum_{i,j}}^\oplus T_{ij}(t), \quad i,j = 1\ {\rm or}\ 2,
\]
with $T_{11}$ and $T_{22}$ the translation unitary group given by \eqref{WAZ} associated with  the -hsm inclusion $\N\subset\M$ w.r.t. $\x$ and $\eta$. 

Since from \eqref{mat} 
\[
\tilde\M_t = \M_t\otimes{\rm Mat}_2(\mathbb C), \quad t \leq R,
\]
and we have $\tilde T(t) \Delta_{\th,\tilde\M} \tilde T(-t) = \Delta_{\th,\tilde\M_t}$, we get
\ben\label{tmt}
T(t)\M T(-t) = \M_t\quad  {\rm and}\quad 
T(t)\Delta_{\eta,\xi,\M}T(-t) = \Delta_{\eta,\xi,\M_t}, \quad t \leq R ,
\een
with $T = T_{21}$. 

Finally, from eq. \eqref{Ht}, we have
\ben\label{H21}
\log\Delta_{\eta,\xi, \N} - \log\Delta_{\eta,\xi} = 2\pi  H\ ,
\een
where $H=H_{21}$, the  restriction of $\tilde H$ to $\H_{21}$, is the generator of $T$. 
\end{proof}
We put now
\[
S(t) \equiv S(\f_t |\!| \psi_t )\ ,
\]
with $\f_t \equiv \f\vert_{\M_t}$, $\psi_t \equiv \psi\vert_{\M_t}$. 
Clearly $S(t)\geq 0$ and, by the monotonicity of the relative entropy, $S(t)$ is a non-decreasing function, thus $S$ is almost everywhere 
differentiable and
\[
S'(t) \leq 0\ .
\]
We are interested in the behaviour of the function $S(t)$. 
\begin{corollary}
In the above setting,
\ben\label{deltaS}
S(t)  = S(\f |\!| \psi) -2\pi t (\xi, H \xi), \quad  0\leq t \leq R \ .
\een
\end{corollary}
\begin{proof}
By \eqref{tmt} and \eqref{cr}, we have
\[
\log\Delta_{\eta,\xi, \M_t} - \log\Delta_{\eta,\xi} = 2\pi t H, \quad  0\leq t \leq R \ .
\]
By taking expectation values on the vector $\xi$, we immediately get
\[
-(\xi, \log\Delta_{\eta,\xi \M_t}\xi)+ (\xi,\log\Delta_{\eta,\xi}, \xi) = 2\pi t (\xi, H \xi), \quad  0\leq t \leq R,
\]
thus \eqref{deltaS} holds. 
\end{proof}
As a first example for a state $\psi$ with the above property \eqref{ws}, let $\psi = \f(U\cdot U^*)$ with $U$ a unitary in $\M_R$. Then
\[
(D\psi : D\f)_s = U\s^\f_s(U^*)
\]
so $(D\psi : D\f)_s \in \M_R$, $s\leq 0$. 

If $\M$ is a type $III_1$ factor, it can be shown that, given any faithful normal state $\psi$ of $\M$ and $R>0$, there exists a normal faithful state $\psi'$ on $\M$ such that $\psi'|_{\M_R} = \psi|_{\M_R}$ and $(D\psi' : D\f)_s \in \M_R$, $s\leq 0$ (this applies to the Remark \ref{rem}). 
\section{Relative entropy in QFT}
This section contains results in local Quantum Field Theory. We first briefly recall the operator algebraic context and we then proceed to describe the relative entropy behaviour. General references for background concepts in the following are \cite{H,OP,BR}.
\subsection{Preliminaries}
We first recall  some of the basics in the operator algebraic approach to QFT. 
\subsubsection{Local von Neumann algebras}
Let $\A$ be a local QFT net of von Neumann algebra on a Hilbert space $\H$ (any spacetime dimension $1 + n$).  Thus,  with $\K$ the family of double cone regions of the Minkowski spacetime, if $O\in \K$ we have a von Neumann algebra $\A(O)$ acting on $\H$. The map 
\[
O\in \K\mapsto \A(O)
\] 
is isotonous, local and Poincar\'e covariant, namely there exists a positive energy, unitary representation of the Poincar\'e group $\Poi$ on $\H$ such that
\[
U(g)\A(O)U(g)^* = \A(gO), \quad g\in\Poi \ , \ O\in\K\ .
\]
Moreover, there exists a unique, up to a phase, unit  $U$-invariant vector $\xi\in\H$, the vacuum vector, and $\xi$ is cyclic for the algebra generated by all the $\A(O)$'s, $O\in \K$. 

With $F$  a spacetime region, let $\mathfrak A(F)$ be the $C^*$-algebra generated by all the von Neumann algebras $\A(O)$ where $O$ runs on the double cones contained in $F$; when the region is the entire Minkowski spacetime, $\mathfrak A \equiv \mathfrak A(\mathbb R^{1+n})$ is called the $C^*$-algebra of $\A$. Denote by $\A(F)$ the weak closure $\mathfrak A(F)''$ of $\mathfrak A(F)$. We assume weak additivity, so that, by the Reeh-Schlieder theorem, $\xi$ is cyclic and separating for $\A(F)$ if $F$ and its spacelike complement $F'$ have non-empty interiors. 

We also assume the Bisognano-Wichmann property, that can be proven under very general assumptions \cite{BW}, in particular we have
\[
\Delta^{-is}_W = U\big(\Lambda_W(2\pi s)\big)\ .
\]
Here $W$ is the wedge region $x_1 > |x_0|$ of the spacetime, $\Delta_W$ is the modular operator associated with $\A(W)$ and $\xi$, and $\Lambda_W$ is the one-parameter group of boosts preserving $W$ (the same property then holds for any other wedge by Poincar\'e covariance). 

Haag duality for wedges then follows:
\[
A(W)' = \A(W')
\]
and we assume Haag duality for double cones
\[
\A(O)' = \A(O'), \quad O\in \K
\]
(otherwise, just redefine $\A(O)$ as $\A(O')'\,$) and similarly for spacelike cones. 

It turns out (see \cite{L82}) that  $\A(W)$ is a factor of type $III_1$ in Connes' classification \cite{C73}. 

\subsubsection{DHR charges}
\label{charges}
The DHR theory was developed in \cite{DHR}, yet we illustrate the charge, or sector, concept also by the later work in \cite{BF}. Let $\r$ be a covariant representation of $\mathfrak A$ on a Hilbert space $\H_\r$, so there exists a positive energy unitary representation  $U_\r$ of the universal cover $\uPoi$ of $\Poi$ on $\H_\r$ such that
\ben\label{covar}
\r(U(g)XU(g)^*) = U_\r(g)\r(X)U_\r(g)^*, \quad X\in \mathfrak A, \  g\in \uPoi .
\een
Assume for the moment that $U_\r$ is massive, namely the energy-momentum spectrum has an isolated lower mass shell. Then, for any spacelike cone $\cal S$ in the Minkowski spacetime, the restriction $\r|_{\mathfrak A(\cal S')}$ is unitarily equivalent to ${\rm id}|_{\mathfrak A(\cal S')}$, with ${\rm id}$ the vacuum representation. Thus, up to unitary equivalence, we may choose a spacelike cone ${\cal S}_0$,  identify $\H_\r$ with $\H$ and assume that $\r(X) = X$, $X\in \mathfrak A({\cal S}_0')$. We then say that $\r$ is localised in ${\cal S}_0$. By Haag duality (for spacelike cones), then $\r$ maps $\mathfrak A({\cal S}_0)$ to $\A({\cal S}_0)$.
Now, $\r|_{\mathfrak A({\cal S}_0)}$ is normal because ${\cal S}_0 \subset \cal S'$ for some spacelike con $\cal S$. So $\r|_{\mathfrak A({\cal S}_0)}$ extends to a normal endomorphism $\r_{{\cal S}_0}$ of $\A({\cal S}_0)$, and similarly to a normal endomorphism $\r_W$ of $\A(W)$ if $W\supset{\cal S}_0$ is a wedge. We may loosely say that $\r_{{\cal S}_0}$ and $\r_W$ are the restrictions of $\r$ to $\A({\cal S}_0)$ and $\A(W)$ and still denote them simply by $\r$ if it is clear from the context that we are dealing with restrictions. 

Moreover there exists a conjugate representation, and this is equivalent to say that $\r$ has finite Jones index; the DHR dimension $d(\r)$ of $\r$ turns out to be the square root Jones index of $\r_W$ \cite{L89, L90}. 

Now, a DHR charge (or sector) is the unitary equivalence class of a representation $\r$ of $\mathfrak A$ such that $\r|_{\mathfrak A(O')}$ is equivalent to ${\rm id}|_{\mathfrak A(O')}$ for all $O\in\K$. By the above arguments, for any given $O\in\K$, we can choose in this class an endomorphism $\r$ of $\mathfrak A$ that is localised in $O$, namely $\r$ acts identically on $\mathfrak A(O')$. We shall thus deal directly with localised endomorphisms. We assume that $\r$ is Poincar\'e covariant (this follows by other requirements, see \cite{GL92}). Also, we assume that $\r$ has finite Jones' index.

Thus there exists a standard left inverse $\Phi$ of $\r$, namely a completely positive map $\Phi : \mathfrak A \to \mathfrak A$ such that
\[
\Phi \cdot\r = {\rm id}\ ,
\]
indeed $\Phi_\r = \r^{-1}\cdot \e$ with $\e$ the minimal conditional expectation $\mathfrak A \to \r(\mathfrak A)$ \cite{L90}.  

Given an endomorphism $\r$ localised in $O$ as above, we shall consider the charged state $\psi$ given by 
\ben\label{psi}
\psi = \f\cdot \Phi \ .
\een
Note that $\psi$ is localised in $O$, namely $\psi = \f$ on $\mathfrak A(O')$, and that, by composing $\psi$ with the adjoint action of a localised unitary, we get a state localised in any given double cone.
If $W$ is a wedge region containing $O$, then $\psi|_{\mathfrak A(W)}$ extends to a normal faithful state of $\A(W)$ ($\psi$ is inner automorphism equivalent to a state localised in $W'$) that we denote by $\psi_W$, and similarly $\f_W = \f |_{\A(W)}$.  .

Let now $u^\r_g \equiv U_\r(g)U(g)^*$, $ g\in \uPoi$, be the covariance cocycle of $\r$. Thus
\[
\r = u^\r_g\r_g(\cdot){u^\r_g}^* \ ,
\]
where $\r_g = U(g)\r\big(U(g)^*\cdot U(g)\big)U(g)^*$; if $\r$ is localised in $O$ then the charge $\r_g$ is localised in $gO$. If $\tilde O$ is a double cone containing both $O$ and $gO$, then both $\r$ and $\r_g$ act identically on $\mathfrak A(\tilde O')$, thus 
\ben\label{loc}
u^\r_g \in \A(\tilde O) \quad {\rm if}\quad \tilde O \supset O\cup O_g 
\een
as, by Haag duality, $u^\r_g \in \A(\tilde O')' = \A(\tilde O)$. 

With $W$ a wedge region and $\Lambda_W$ the  corresponding boost one parameter group, let $\r$ be localised in a (possibly unbounded) region contained in $W$. Then $u^\r_{\Lambda_W(s)}\in \A(W)$, $s\in \mathbb R$. We shall make use of the following relation with Connes' cocycle \cite{L97}:
\ben\label{KW}
u^\r_{\Lambda_W(-2\pi s)} = d(\r)^{is}(D\psi_W : D\f_W)_s, \quad s\in \mathbb R \ .
\een
Thus, while the Bisognano-Wichmann theorem sets up a connection between modular theory and the vacuum boost symmetries, formula \eqref{KW} sets up a connection between the relative modular operator and the boost symmetries in the charged representation; this formula is indeed equivalent to the equality\ben\label{KW1}
2\pi K_{\r,W} =- \log \Delta_{\eta,\xi, W} - \log d(\r) \ ,
\een
therefore, evaluating on $\xi$,
\ben\label{sfp}
S(\f_W |\!| \psi_W) = 2\pi (\xi, K_{\r,W}\xi) + \log d(\r)   \ .
\een
Here $K_{\r,W}$ is the selfadjoint generator of the boost one parameter unitary group $U_\r(\Lambda_W(\cdot))$,
 $\xi$ is the vacuum vector, $\eta$ is any cyclic and separating vector for $\A(W)$ giving the state $\psi_W$ and $\Delta_{\eta,\xi, W}$ is the relative modular operator for $\A(W)$ associated with $\xi,\eta$. 

Formula \eqref{KW1} is valid only if $\r$ is localised within $W$. We shall move the wedge $W$ in formula \eqref{KW1} to another wedge $\bar W$ still containing the localisation region of $\r$. After that, we will be able to compare $\log \Delta_{\eta,\xi, W}$ and $\log \Delta_{\eta,\xi, \bar W}$, hence the relative entropies, by comparing $K_{\r,W}$ and $K_{\r,\bar W}$ by the Poincar\'e unitary action $U_\r$. 

For the direct relation to the QNEC, we also give now the formula\footnote{If we choose the vector representatives $\eta$ of $\f_W$ in the natural cone $\cal P_\xi^\natural \subset \H$ relative to $\A(W)$ (see \cite{T, BR}), then
$\Delta_{\xi,\eta, W} = \Delta_{\eta,\xi, W}^{-1}$; thus, by evaluating \eqref{KW1} on $\eta$, we also have the formula
\[
S(\psi_W |\!| \f_W) = - 2\pi (\eta, K_{\r,W}\eta) - \log d(\r)   \ .
\]
However, by restricting to a von Neumann subalgebra, $\eta$ is no longer canonical.}
for $S(\psi_W |\!| \f_W)$. 
\begin{lemma}
\label{foot}
Let $\r$ be localised in $W$ and choose a conjugate charge $\bar\r$ localised in $W'$. There exists a vector $\eta\in \H$, giving the state $\psi_W$, such that
\ben\label{lemma}
S(\psi_W|\!|\f_W) = 2\pi(\eta,  K_{\bar\r, W}\eta)  -\log d(\r)\ .
\een
\end{lemma}
\begin{proof}
We omit here the suffix $W$. First we give a formal argument that does not take into account operator domains. Since  
\[
(D\psi : D\f)_s = (D\f : D\psi)^*_s \ ,
\]
by differentiating at zero the adjoint of both sides in \eqref{KW}, we have by eq. \eqref{w*} that
\ben\label{kk}
2\pi(K_\r - K_0) = \log \Delta_{\xi,\eta} -  \log\Delta_\eta    -\log d(\r) \ .
\een
Let $\bar \r$ be a conjugate charge localised in $W'$, say $\bar\r = j\cdot\r\cdot j$ with $j ={\rm Ad}J$ the adjoint action of the modular conjugation $J$ of $\A(W),\xi$ \cite{GL92}. Since the boost covariance unitary cocycles $u^\r$ and $u^{\bar \r}$ are localised in spacelike separated regions, we have 
$u^{\bar\r\r} = \r(u^{\bar\r})u^\r =  u^{\bar\r}u^\r$ (see \eqref{ur} below), thus the following relation among the selfadjoint boost generators in the vacuum and in the charged representations holds:
\[
 K_{\r} -  K_0 =  K_{\bar\r\r} - K_{\bar\r}   \ .
\]
Let $V$ be a canonical isometry that intertwines the vacuum and the $\bar\r\r$ representations, namely 
\ben\label{V}
VX =\bar\r\r(X)V, \quad X\in\mathfrak A ,
\een
(see \cite{LRob}; if $\r$ is irreducible, $V$ is uniquely determined by \eqref{V} up to a phase). The vector $\eta \equiv V\xi$ gives the state $\psi$ \eqref{psi} on $\A(W)$; indeed $\Phi = V^*\bar\r(\cdot)V$ is the minimal left inverse on $\A(W)$, so
\[
\psi(X) = (\xi,  V^*\bar\r(X)V\xi) = (\eta, \bar\r(X)\eta) = (\eta, X\eta), \quad X\in \A(W)  ,
\]
because $\bar\r$ acts identically on $\A(W)$.   

Moreover $\eta$ is cyclic for $\A(W)$ ($VV^*$ is the Jones projection for the inclusion $\r(\A(W))\subset \A(W)$ \cite{L90}). 

By evaluating both sides of \eqref{kk} on $\eta$, we conclude that
\[
S(\psi|\!|\f) = 2\pi(\eta,  K_{\bar\r}\eta) - \log d(\r)\ .
\]
Indeed,
\begin{multline*}
S(\psi|\!|\f)  = - (\eta, \log\Delta_{\xi,\eta}\eta)  =  (\eta, ( \log\Delta_{\eta} - \log\Delta_{\xi,\eta})\eta) 
= 2\pi(\eta, (K_0 - K_\r )\eta)  - \log d(\r) \\
=  2\pi(\eta, (K_{\bar\r} - K_{\bar\r\r})\eta) - \log d(\r) = 2\pi(\eta,  K_{\bar\r}\eta) - \log d(\r) 
\end{multline*}
because
\[
(\eta, K_{\bar\r\r} \eta) = (V\xi, K_{\bar\r\r} V\xi) = (\xi , K_0 \xi) = 0\ ,
\]
since $VK_0 =  K_{\bar\r\r} V$. 

More rigorously, by \eqref{w*} and \eqref{KW}, setting $u^\r_s \equiv u^\r_{\Lambda_W( s)}$,
we have 
\begin{multline*}
S(\psi|\!|\f)  = -i\frac{\rm d}{{\rm d}s}\psi(w_s)|_{s=0} = i\frac{\rm d}{{\rm d}s}d(\r)^{is}\psi(u^\r_{2\pi s})\big|_{s=0}
= i\frac{\rm d}{{\rm d}s}d(\r)^{is}(\eta, u^\r_{2\pi s}\eta)\big|_{s=0}\\
= i\frac{\rm d}{{\rm d}s}d(\r)^{is}(V\xi, {u^{\bar\r\ *}_{2\pi s}} u^{\bar\r \r}_{2\pi s} V\xi)\big|_{s=0}
= i\frac{\rm d}{{\rm d}s}d(\r)^{is}(\eta, {u^{\bar\r\ *}_{2\pi s}}  \eta)\big|_{s=0}\\
= -i\frac{\rm d}{{\rm d}s}d(\r)^{-is}(\eta, {u^{\bar\r}_{2\pi s}}  \eta)\big|_{s=0}= 2\pi(\eta, K_{\bar\r} \eta) - \log d(\r)
\end{multline*}
because
\ben\label{ur}
u^{\bar\r \r}_s = \r(u^{\bar\r}_s) u^{\r}_s =u^{\bar\r}_s u^{\r}_s
\een
and
\[
u^{\bar\r \r}_s V = V
\]
by the two-variable tensor categorical cocycle property of $u^\r_s$ and \eqref{V}, \cite{L97}. 
\end{proof}
Note that, in the above proof, the isometry $V$ does not depend on $W$, due to eq. \eqref{V}. So  \eqref{lemma} holds true for all wedges $W$ such that $\r$ is localised in $W$ and $\bar\r$ is localised in $W'$. 
\begin{theorem}\label{eta}
Let $\r$ be localised in $\cal S$ and $\eta\in\H$ a vector which is cyclic and gives the state $\psi$ on $\A({\cal S}_1')$. Here $\cal S$ and ${\cal S}_1$ may be double cones, spacelike cones or wedges and are causally disjoint. 

There exists a conjugate charge $\bar\r$, localised in ${\cal S}_1'$, such that
\[
S(\psi_W |\!|\f_W) = 2\pi(\eta,  K_{\bar\r,W}\eta)  -\log d(\r)\ ,
\]
for every wedge $W$ with ${\cal S} \subset W\subset {\cal S}_1'$. 
\end{theorem}
\begin{proof}
Choose a conjugate charge $\bar \r$ localised in ${\cal S}_1$. Then the statement holds with $\eta = V\xi$ given in the proof of Lemma \ref{foot}. 
So it is sufficient to show that every vector $\eta'$ giving the state $\psi$ on $\A({\cal S}_1')$ is of this form. Indeed $\eta' = U\eta$ with $U$ a unitary in $\A({\cal S}_1)$, so $\eta_1 = V_1\xi$ with $V_1 = UV$ the isometry that intertwines the identity and $\bar\r_1\r$ as in the proof of the lemma, where $\bar\r_1 = U\bar\r(\cdot) U^*$ is a conjugate of $\r$ localised in ${\cal S}_1$. 
\end{proof}
Due to the above theorem, the analyses in the following for $S(\f_W |\!|\psi_W)$ have analogue corresponding analyses for $S(\psi_W |\!|\f_W)$. 

\subsection{Relative entropy and geometric deformations}\label{egd}
We now study the behaviour of the relative entropy associated with a local von Neumann algebras by varying the reference spacetime region. 
\subsubsection{Constant spatial shifts}
Let $\A$ be a local QFT net on the Minkowski spacetime $\mathbb R^{1+n}$. 
With $W$ the wedge $x_1 > |x_0|$ and ${\bf z}\in W$ a point in $W$, we denote by $W_{\bf z}\subset W$ the
subwedge   $W_{\bf z} \equiv W + {\bf z}$. Set ${\bf u} = (1,1,0\dots,0)$ as above and ${\bf v} = (1,-1,0\dots,0)$ and consider two points ${\bf z_k} = a_k{\bf u} - b_k{\bf v}$, $k = 1,2$, with $a_k,b_k\geq 0$, thus ${\bf z_k}\in W$ and set
\[
O_{\bf z_1, z_2} \equiv W_{\bf z_1} \cap W'_{\bf z_2}
\]
if ${\bf z_2}\in W_{\bf z_1}$, namely $a_2\geq a_1, b_2\geq b_1$. Note that the local von Neumann algebra $\A(O_{\bf z_1, z_2})$ is cyclic on the the vacuum vector by the Reeh-Schlieder theorem if $a_2> a_1, b_2> b_1$. 
\paragraph{Case of one charge.}\label{1charge}

Let's consider the case of one DHR charge $\r$ localised in a region contained in $O_{\bf z_1, z_2}$. 

As above, let $\f$ be the vacuum state on $\mathfrak A$ and $\psi = \f\cdot\Phi$ the charged state, where $\Phi$ is the left inverse of $\r$. Obviously, $\f|_{{\mathfrak A}(W_{\bf z})}$ extends to a faithful normal state on $\A(W_{\bf z})$, that we denote by $\f_{\bf z}$. As $\Phi$ is normal on ${\mathfrak A}(W_{\bf z})$, also $\psi|_{{\mathfrak A}(W_{\bf z})}$ extends to a faithful normal state on $\A(W_{\bf z})$, that we denote by $\psi_{\bf z}$. 

We set 
\[
S({\bf z})  \equiv S(\f_{\bf z} |\!| \psi_{\bf z}), \quad {\bf z}\in W\ .
\]
Let $U_\r$ be the covariance positive energy unitary representation of $\uPoi$ associated with $\r$ \eqref{covar} and denote by 
$H_{\r,+}$ and $H_{\r,-}$ the positive generators of the light-like translation unitary one parameter groups given by $U_\r$ in the $\bf u$ and $\bf v$ direction.  
\begin{proposition}\label{stP}
If $\r$ is localised in $O_{\bf z_1, z_2}$ as above, then
\ben\label{st}
S({\bf z}) = S(\f_{\bf 0} |\!| \psi_{\bf 0}) -  2\pi \big(a(\xi, H_{\r_+} \xi) + b(\xi, H_{\r_-} \xi)\big)\ ,
 \quad {\bf z} \in O_{\bf 0, z_1}\  ,
\een
and
\ben\label{st2}
S({\bf z}) = 0, \quad {\bf z} \in W_{\bf z_2} \ ,
\een
where ${\bf z} = a{\bf u }- b\bf v$. 
\end{proposition}
\begin{proof}
We write
\[
S({\bf z}) - S(0) = \big(S(a{\bf u} - b{\bf v}) - S(a{\bf u})\big) + \big(S(a{\bf u}) - S(0)\big) \ .
\]
Now $S(a{\bf u}) - S(0) = 2\pi a (\xi, H_{\r_+} \xi)$ by Theorem \ref{main}. The inclusion $\A(W_{a{\bf u}})\subset \A(W_{a{\bf u} - b{\bf v}})$ is +half-sided modular w.r.t. $\xi$ and the associated translation Hamiltonian is $H_{\r_-}$ because it is the -lightlike translation Hamiltonian in the representation $\r$, thus the equation \eqref{st} follows. 
\end{proof}
We make explicit the space-translation case where $\bf z_1, z_2$ belong to the time-zero hyperplane.  Thus $a = b$, ${\bf z_1} = (0, R,0\dots,0)$, ${\bf z_2} = (0, \tilde R,0\dots,0)$ and we set  $W_t = W_{(0,t,0\dots,0)}$, $t\geq 0$. 

Setting in \eqref{st} $a= b =   \frac{t}{\sqrt2}$ we have
\[
S(t) \equiv S(\f_{t} |\!| \psi_{t}) = S(\f_{\bf 0} |\!| \psi_{\bf 0}) -  2\pi t(\xi, H_{\r} \xi) , \quad t \leq R,
\]
\[
S(t) = 0, \quad t\geq \tilde R,
\]
with $\f_t , \psi_t$ the restrictions of $\f,\psi$ to $\A(W_t)$. 
Here $H_\r = \frac{1}{\sqrt2}(H_{\r_+} +  H_{\r_-})$ is  the Hamiltonian, i.e. the generator of the time translation unitary one parameter group in the representation $\r$, and $t$ may be also negative as the reference wedge $W$ could be replaced with $W_{t'}$ for any $t'\leq 0$. 

Now, as shown in \cite{L97}, we have
\ben\label{S0}
S(0) = 2\pi (\xi, K_\r\,\xi) + \log d(\r)\ ,
\een
thus
\ben\label{st1}
S(t)  = 2\pi (E_{\rm loc} -   tE) + \log d(\r)  , \quad t \leq R,
\een
where $E_{\rm loc} \equiv (\xi, K_\r\,\xi)$ is the is mean vacuum energy for the Rindler observer and $E$ is mean vacuum energy in the rest frame.
Notice that the derivative 
\[
S'(t) = -2\pi E, \quad E \equiv (\xi, H_{\r} \xi) ,
\]
is independent of the spatial direction as it is related to $E$ and not to $E_{\rm loc}$.
Indeed, the asymptotic mean entropy in any direction is
\ben\label{end}
\lim_{t\to+\infty}\frac{S(-t) - S(t)}{2t} = \pi E \ ,
\een
which is related to Bekenstein's bound (cf. \eqref{end2}).  

\begin{remark}
By our arguments, one can also study the dependence of $S(\f_{gW} |\!| \psi_{gW})$ as $g\in\Poi$ varies so that $gW$ contains a fixed spacelike cone $\cal S$ with of $\r$ localised in $\cal S$. For example, one may rotate $W$ by a suitably small angle. 
\end{remark}
\begin{remark}\label{rem}
In Proposition \ref{stP}, one may take in particular $\r$ to be the inner automorphism implemented by a unitary $U\in\A(O_{{\bf z_1} , {\bf z_2}})$.  Then 
\[
\psi = \f_U, \quad \f_U(X) = \f(U^*XU), \quad X \in \mathfrak A \ .
\] 
In this case $\psi$ belongs to the same folium of the vacuum state, indeed $\psi$ is the expectation value on the vector $\eta = U\xi$ in the vacuum Hilbert space, $\psi(X) = (\eta, X\eta)$. In other words, the representation $\r$ of $\mathfrak A$ is unitarily equivalent to the vacuum representation, so it has zero charge. 

A state $\psi = (\eta, \cdot\, \eta)$ on $\mathfrak A$ is of the form $\f_U$ with $U$ a unitary in $\mathfrak A(O)$, $O\in\K$, iff $\psi|_{\mathfrak A(O')} = \f|_{\mathfrak A(O')}$ and $\eta$ cyclic for $\A(O)$ (if $\eta$ is not cyclic then $U$ is an isometry). Since, in this case, $K_\r = UKU^*$ and $H_\r = UHU^*$, formula \eqref{st1} reads
\[
S(t) = 2\pi (\eta, K\eta) -   t(\eta, H\eta)  , \quad t \leq R,
\]
with $H$ and $K$ the Hamiltonian and the Rindler Hamiltonian in the vacuum sector.

This class of localised states is large. If $\psi$ is any faithful normal state of $\A(W)$, given $\e >0$ there exists a unitary $U\in \A(O)$ for some double cone $O\subset W$ such that
\[
||(\psi - \f_U )|_{\A(W)}||<\e \ .
\]
This follows because $\A(W)$ is a factor of type $III_1$, see \cite{L82}, 
these unitaries form a dense set in the unitary group of $\A(W)$ and, in a $III_1$-factor, by Connes-St\o rmer's theorem \cite{CS}, the orbit of any faithful normal state under inner automorphisms is norm dense in the set of all faithful normal states. 
\end{remark}
\paragraph{Case of multiple charges.}

We now describe the space translation case in presence of multiple charges when the localisation regions are the casual envelop of time-zero regions, the discussion in a more general case as in Section \ref{1charge} follows the same lines.  

Suppose we have charges $\r_1 ,\r_2,\dots \r_\nu$ with $\r_k$ localised in 
$O_{R_i , \tilde R_i}$, where $R_{i+1} > \tilde R_i > R_i$ and $O_{R_i , \tilde R_i}= W_{R_i}\cap W'_{\tilde R_i}$. As the $\r_i$'s are localised in spacelike separated regions, they mutually commute. Let 
\[
\r = \r_1\r_{2}\cdots\r_\nu 
\]
be the composition of all the charges $\r_k$. Let $\Phi$, $\Phi_k $ be the (minimal) left inverses of $\r$, $\r_k$ and $\psi = \f\cdot\Phi$,  $\psi_k = \f\cdot\Phi_k$ that we consider as states of $\A(W)$.  
\begin{proposition}\label{additivity}
\[
S(\f|\!| \psi) = \sum_{k=1}^\nu S(\f_k|\!| \psi) \ .
\]
\end{proposition}
\begin{proof}
 Let $u^{\r}$ be the covariance unitary cocycle for the boost action associated with $\r$, namely
$u^{\r}_s = U_\r\big(\Lambda_W( s)\big)U\big(\Lambda_W(-s)\big)$, and similarly for $u^{\r_k}$. 

Let $\eta$ and $\eta_k$ be cyclic and separating vectors on $\H$ giving the states $\psi$ and $\psi_k$ on $\A(W)$. 
We have 
\begin{align}
u^{\r}_{-2\pi s} &=  d(\r)^{is}(D\psi : D\f)_{s} = d(\r)^{is}\Delta_{\eta,\xi}^{is}\Delta_\xi^{-is}\ , \\
u^{\r_k}_{-2\pi s} &=  d(\r_k)^{is}(D\psi_k : D\f)_{s} = d(\r_k)^{is}\Delta_{\eta_k,\xi}^{is}\Delta_\xi^{-is}\ .
\end{align}
For simplicity set now $\nu=2$.
 For small $s,s'$,  $u^{\r_1}_s$ and $u^{\r_2}_{s'}$ are localised in spacelike separated regions, in particular they each other commute. By \cite[Prop. 1.4]{L97} we have
\[
u^{\r}_s  = u^{\r_1\r_2}_s = \r_1(u^{\r_2}_s) u^{\r_1}_s = u^{\r_2}_s u^{\r_1}_s
\]
because $\r_1$ acts identically on $u^{\r_2}_s$ by spacelike separation for small $s$. Thus
\[
\frac{\rm d}{{\rm d}s} u^{\r}_s \big|_{s=0} = \frac{\rm d}{{\rm d}s} u^{\r_1}_s u^{\r_2}_s\big|_{s=0} =
\frac{\rm d}{{\rm d}s} u^{\r_1}_s \big|_{s=0} + \frac{\rm d}{{\rm d}s} u^{\r_2}_s \big|_{s=0} \ .
\]
Therefore
\begin{align*}
S(\f |\!| \psi) &= i\frac{\rm d}{{\rm d}s}(\xi, (D\psi : D\f)_s\,\xi)\big|_{s=0} \\
&= i\frac{\rm d}{{\rm d}s}(\xi, u^\r_{-2\pi s}\xi)\big|_{s=0} + \log d(\r)\\
&= i\frac{\rm d}{{\rm d}s}(\xi, u^{\r_1}_{-2\pi s}\xi)|_{s=0} + i\frac{\rm d}{{\rm d}s}(\xi, u^{\r_2}_{-2\pi s}\xi)\big|_{s=0}+ \log d(\r)\\
&= i\frac{\rm d}{{\rm d}s}(\xi, u^{\r_1}_{-2\pi s}\xi)|_{s=0} + i\frac{\rm d}{{\rm d}s}(\xi, u^{\r_2}_{-2\pi s}\xi)\big|_{s=0}+ \log  d(\r_1 ) + \log d(\r_2)\\
&= i\frac{\rm d}{{\rm d}s}(\xi, (D\psi_1 : D\f)_{s}\,\xi)|_{s=0} + i\frac{\rm d}{{\rm d}s}(\xi, (D\psi_2 : D\f)_{s}\,\xi)\big|_{s=0}\\
&= S(\f |\!| \psi_1) + S(\f |\!| \psi_2)
\end{align*}
by the multiplicativity of the minimal dimension. 
\end{proof}
Since $\r_k$ acts identically on $\A(O_t)$ if $t \leq  R_k$, $O_t \equiv  W'_t\cap W$,
we have
\[
\psi_1 \big| _{\A(W_t)} = \psi_k \big| _{\A(W_t)}, \quad t\geq \tilde R_k \ .
\]
By Proposition \ref{additivity} we have
\[
S(t) \equiv S(\f_t |\!| \psi_t) =  \sum_{j=k}^n S({\f_k}_t |\!| \psi_t), \quad \tilde R_{k-1}\leq t \leq R_k \ .
\]
By the discussion in the previous  paragraph, we then have (see Fig. 1):
\begin{figure}
\centering
\begin{tikzpicture} 
\draw [magenta,very thick,dotted](0.5,3) to[out=-15,in=170] (1.2,2.5);
\draw [magenta,very thick,dotted](2,2.5) to[out=-15,in=170] (3.3,1.2);
\draw [magenta,very thick,dotted](5.1,1.2) to[out=-15,in=170] (7,0.7);
\draw [magenta,very thick,dotted](7.7,0.7) to[out=-15,in=170] (8.5,0);
\draw [->] (-0.3,0) -- (9,0); 
\draw [->] (0,-0.3) -- (0,3.3); 
\draw[magenta,very thick] [-] (0,3) -- (0.5,3); 
\draw[magenta,very thick] [-] (1.2,2.5) -- (2,2.5); 
\draw[very  thin] [-] (0,2.5) -- (1.2,2.5); 
\draw[magenta,very thick] [-] (3.3,1.2) -- (5.1,1.2); 
\draw[very  thin] [-] (0,1.2) -- (3.3,1.2); 
\draw[magenta,very thick] [-] (7,0.7) -- (7.7,0.7); 
\draw[very  thin] [-] (0,0.7) -- (7.7,0.7); 
\draw[magenta,very thick] [-] (8.5,0) -- (9,0); 
\draw[very thick] [-] (0.5,0) -- (1.2,0); 
\draw[very thick] [-] (2,0) -- (3.3,0); 
\draw[very thick] [-] (5.1,0) -- (7,0); 
\draw[very thick] [-] (7.7,0) -- (8.5,0); 
\draw (-0.5,2.76) node {\small $2\pi E_1$};
\draw (-0.5,1.87) node {\small $2\pi E_2$};
\draw (-0.5,0.96) node {\small $2\pi E_3$};
\draw (-0.5,0.38) node {\small $2\pi E_4$};
\draw (0.9,0.2) node {\small $\r_1$};
\draw (2.7,0.2) node {\small $\r_2$};
\draw (6.05,0.2) node {\small $\r_3$};
\draw (8,0.2) node {\small $\r_4$};
\draw[magenta] (4,1.8) node {\small $S'(t)$};
\end{tikzpicture}
\caption{\footnotesize Entropy density plot in a four charge case (in red).  The jump of $S'(t)$ over the support of the charge $\r_k$ (thick segment) is equal to $2\pi$ times the mean energy of $\r_k$. In general, $S''(t)$ is only known in mean on the charge support.}
\end{figure}
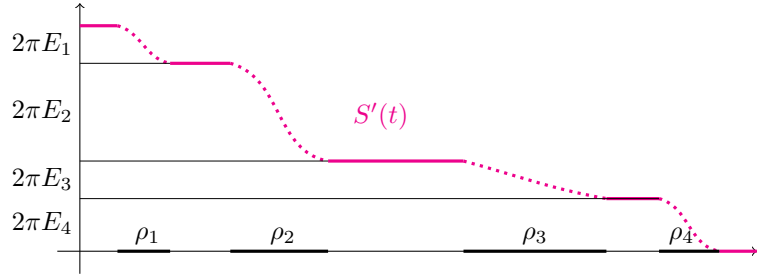
\begin{theorem}
With $E_k \equiv (\xi, H_{\r_k}\xi)$ then mean vacuum energy of the charge $\r_k$, we have
\[
S(t) =   S(0) -2\pi t\sum_{j=k}^\nu E_k, \quad \tilde R_{k-1}\leq t \leq R_k \ ,
\]
$S(t) = S(0) -2\pi t\sum_{j=1}^\nu E_k$ if $t \leq R_1$ and $S(t) = 0$ if $t\geq \tilde R_\nu$. Here $S(0)$ is given by \eqref{S0}. 
\end{theorem}
Thus the  second derivative satisfies
\[
S''(t) = 0, \quad \tilde R_{k-1}\leq t \leq R_k \ .
\]
In the intervals  $(R_k, \tilde R_k)$ the relative entropy $S(t)$ monotonically decreases and $S''$ is positive in the average, namely
\ben\label{Av}
\int_{R_k}^{\tilde R_k}  S''(t) {\rm d} t = 2\pi E_k >0\ .
\een
\paragraph{Conformal case.}
We now assume $\A$ to be conformally covariant, namely the unitary representation $U$ of $\Poi$ extends to a unitary representation on $\H$ of the conformal group, still denoted by $U$, and $\A$ is $U$-covariant. In particular
\[
U(\delta_s)\A(O_1)U(\delta_s)^* = \A(O_{ r}), \quad O\in\K,\ r = e^s\ ,
\]
where $O_r\in\K$ is the causal envelope of the time-zero sphere with center at the origin and radius $r>0$. 

Let ${r_0}>0$ and $\r$ be a (finite index) DHR charge localised in a $O_{r_0}$; we assume that $\r$ is conformally covariant and denote by $U_\r$ the associated covariance unitary representation of the universal cover of the conformal group on $\H$. With $\f$ the vacuum state and $\psi =\f\cdot\Phi$ the charged state as in Sect. \ref{charges}, we set $\f_r = \f|_{\A(O_r)}$, $\psi_r = \psi|_{\A(O_r)}$ and
\[
S(r) = S(\f_r |\!| \psi_r) \ .
\]
Choose $r' > {r_0}$ and let $\eta\in\H$ be a cyclic and separating vector giving the state $\f_{r'}$, thus $\f_r$ for any $0<r\leq r'$. We have
\[
S(r) = -(\xi , \log\Delta_{\eta,\xi, r}\xi)\ ,
\]
with $\Delta_{\eta,\xi,r}$ the relative modular operator with respect to the von Neumann algebra $\A(O_r)$. 

Let $\Lambda_{O_r}$ be the one parameter group of special conformal transformation preserving $O_r$ that is conjugate to the boost one-parameter group $\Lambda_W$, thus \cite{HL}
\[
\Delta_{\xi,r}^{-is} = U\big(\Lambda_{O_r}(2\pi s)\big)\ .
\]
We have 
\ben\label{dk}
 \log\Delta_{\eta,\xi, r} = -2\pi K_{\r, r} - \log d(\r) \ ,
\een
where  $K_{\r, r}$ is the selfadjoint generator of the one-parameter unitary group $U_\r(\Lambda_{O_r}(\cdot))$ \cite{L97}. 

Now, the subgroup of the universal cover of the conformal group generated by time translations, dilations and ray inversion map is  naturally isomorphic to the universal cover $\overline{SL(2,\mathbb R)}$ of $SL(2,\mathbb R)$. With this identification, the lift to $\overline{SL(2,\mathbb R)}$ of the one parameter subgroup $\g$ of $SL(2,\mathbb R)$ \eqref{gs}, still denoted by $\g$, satisfies
\[
\g_s = \Lambda_{O_1}(s)\ ,
\] 
see \cite{BGL}. 
By the Lie algebra relation \eqref{att'}, we then have
\[
K_{\r, 1} = \frac12({H_\r - H'_\r})\ ,
\]
where $H_\r$ is the selfadjoint generator of the time-translation group and $H'_\r$ is conjugate to $H_\r$ by the ray inversion unitary.
Thus, by eq. \eqref{dil}, we have
\[
K_{\r, r} = U_\r(\delta_s)K_{\r, 1}U_\r(\delta_s)^* = \frac12({e^s H_\r - e^{-s}H'_\r})\ , \quad r = e^s\ .
\]  
If $r \geq {r_0}$, it then follows by \eqref{dk} that
\ben\label{dkd}
\log\Delta_{\eta,\xi, r} = -2\pi K_{\r, r} - d(\r) = -\pi({e^s H_\r - e^{-s}H'_\r}) - \log d(\r) \ .
\een
\begin{theorem}
\label{main2} 
If $r\geq {r_0}$, we have
\ben\label{mainf}
S(r) = \pi\Big({r E - \frac1r E'}\Big) + \log d(\r)\ ,
\een
with $E = (\xi, H_\r \xi)$, $E' = (\xi, H'_\r \xi)$ the mean vacuum energies corresponding to the Hamiltonians $H_\r$ and $H'_\r$. 
\end{theorem}
\begin{proof}
Taking expectation values on $\xi$ in formula \eqref{dkd} we get
\[
(\xi, \log\Delta_{\eta,\xi, r}\xi) = -\pi\big(e^s(\xi, { H_\r\xi) - e^{-s}(\xi, H'_\r}\xi)\big) - \log d(\r) \ ,
\]
so the expression \eqref{mainf} for $S(r)$ is valid for ${r_0}\leq r\leq r'$ by the above arguments, hence for all $r\geq r_0$ as $r'$ is arbitrary. 
\end{proof}
Similarly, by Theorem \ref{eta}, we have
\ben\label{sr}
 S(\psi_r |\!| \f_r) =  \pi\Big({r \bar E - \frac1r \bar E'}\Big) - \log d(\r)\ , \quad r\geq r_0\ .
\een
Here we choose $r' > r$, a conjugate charge $\bar \r$ localised in a double cone contained in $O'_{r}\cap O_{r'}$, a vector $\eta$ giving the state $\psi_r$ and set $\bar E = (\eta, H_{\bar \r} \eta)$, $\bar E' = (\eta, H'_{\bar \r} \eta)$. Note that, as a consequence of formula \eqref{sr}, $\bar E$ and $\bar E'$ are independent of this choice. 

Clearly, the role of $\r$ and $\bar \r$ may be interchanged. So we have
\[
 S(\psi_r |\!| \f_r)  \leq \pi r \bar E  \ ,
\]
namely, by adding a charge in the region $O'_r$, the vacuum relative entropy in the region $O'_r$ is dominated by $2\pi r$ the corresponding energy increase in the region $O_r$. This implies Bekenstein's bound. 

Note also that
\[
S'(r) = \pi(E + E'/r^2) \geq 0, \qquad S''(r) = -2\pi E'/r^ 3 \leq 0, \quad r\geq r_0 \ ,
\]
and the asymptotic mean entropy is
\ben\label{end2}
\lim_{r\to\infty} \frac{S(r)}{r} = \pi E 
\een
(cf.  \eqref{end}). 

The symmetry 
\[
S(r) \leftrightarrow S(-1/r),  \quad E\leftrightarrow E'
\]
is a manifestation of the charge conjugation symmetry, see \cite{GL92}. 

An interesting point in Thm. \ref{main2} is the appearance of $\log d(\r)$ as constant term in the expansion of $S(r)$.
A similar occurrence holds in the increment of the asymptotic expansion of total entropy in chiral CFT (Kac-Wakimoto conjecture), see \cite{KL05}. 

Of course, this section can be generalised to the multiple charge case by the same footing as in the previous section.

\subsubsection{General deformations in a null direction and ANEC}\label{null}
Let $\A$ be a local QFT net on the Minkowski spacetime as above. Let $W$ be the wedge region $x_1 > |x_0|$. We use also the coordinates $u = x_0 + x_1$, $v= x_0 - x_1$, $y_k = x_k$, $k > 1$. 

Let $f(y)$ be a non-negative function of $y = (y_2,\dots y_n)$ and $W_f$ the region obtained by $W$ by the transformation $(u,v, y) \to (u + f(y), v , y)$ (see Fig. 2), thus
\ben\label{Wf}
W_f = \{(u,v,y) : u > f(y), v <0\}\ .
\een
\begin{figure}
\centering
\begin{tikzpicture} \label{deform}
\draw [->] (4,0) -- (5,2); 
\draw [->] (-4,0) -- (-3,2); 
\draw [->] (-4,0) -- (-3,-0.5);
\draw [->] (4,0) -- (5,-0.5); 
\draw [orange, thick, domain= 0:1] plot (\x/0.5 - 1, {\x- \x*\x*\x} ); 
\draw[ thick, orange] [-] (-4,-0) -- (-1,0); 
\draw[ thick, orange] [-] (1,0) -- (4,0); 
\draw[ dotted,thin] [-] (-1,0) -- (1,0); 
\draw
    (5.5,2) node {$u$};
    \draw
 (5.5,-0.5) node {$v$};
       \draw[orange]
 (4.4,1.7) node {$A$};
     \draw[orange]
 (4.2,-0.5) node {$A$};
     \draw
 (0.06,0.1) node { $f$};
  \draw
 (0.34,1.2) node { $f_t$};
 \draw [orange, thick, dashed, domain= 0:1] plot (\x/0.5 - 1, {1.6*\x- 1.6*\x*\x*\x} ); 
 \draw [orange, thick, dashed, domain= 0:1] plot (\x/0.5 - 1, {2.2*\x- 2.2*\x*\x*\x} ); 
\end{tikzpicture}
\caption{\footnotesize The function $f$ is the boundary of the deformed region on the null horizon. The entire deformed region is its causal envelop $A$. }
\end{figure}
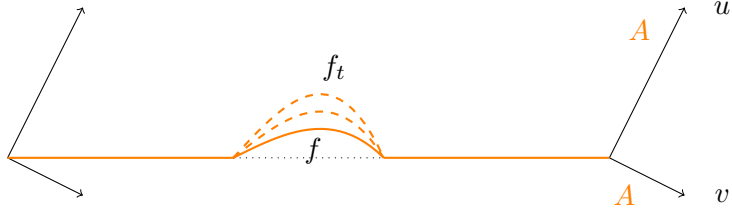
Consider the von Neumann algebras $\M= \A(W)$ and $\M_f = \A(W_f)$ and note that $\M_f\subset \M$ is -hsm with respect to the vacuum vector $\xi$. So we have the tunnel of von Neumann algebras $\M_t$ as in \eqref{tt} with $\N = \M_f$
\ben\label{tun}
\M_t = \s^\f_{s}(\M_f)  = \A(W_{f_t} )\ ,
\een
$ t =   e^{-2\pi s} - 1$, $f_t \equiv e^{-2\pi s} f$. 

The one parameter unitary group $T_f$ generated by the associated Hamiltonian 
\[
H_f = \frac{1}{2\pi}\big(\log\Delta_{\xi, \A(W_f)} - \log\Delta_{\xi, \A(W)}\big)
\]
satisfies
\[
T_f(t)\A(W)T_f(-t) = \A(W_{f_t}), \quad t >0,
\]
namely $T_f(t)\A(W)T_f(-t) = \M_{t}$, in particular $\A(W_f) = T_f(1)\A(W)T_f(-1)$. 

Let $\f$ be the vacuum state and $\psi$ the state obtained by $\f$ by adding a DHR charge $\r$ localised in a double cone $O\subset \M_R$ with $R >0$ as in \eqref{psi}, and let $\eta$ be a cyclic and separating vector giving $\psi$ on $\M$. 

Let $u^\r$ be the covariance unitary cocycle for $\r$. 
If $\Lambda(s) : (u,v, y)\mapsto (e^s u,e^{-s}v, y)$
is the $s$-boost  in the $x_1$-direction, $s\in\mathbb R$, we simply set
$u^\r_s = u^\r_{\Lambda(s)}$ as above. 

Since $\r$ is localised in $W_{f_R}$ and $W_{f_R}$ is mapped into itself by $\Lambda(s)$, with $s\geq 0$,
by \eqref{loc} we have $u^\r_s\in \M_R\equiv \A(W_{f_R})$ if $s\geq 0$. 

Now, by \eqref{KW}, we have
\[
u^\r_{-2\pi s }= d(\r)^{is}(D\psi : D\f)_s \ ,
\]
so $(D\psi : D\f)_s\in \M_R$, $s\leq 0$; here $\f$ and $\psi$ are considered as states on $\M =\A(W)$. 
Therefore we are in the setting of Section \ref{Wies}, with $\N = \A(W_f)$, and we can apply Thm. \ref{rw}. 
We have
\ben
\log \Delta_{\eta, \xi, \M_t} - \log \Delta_{\eta, \xi, \M}  = 2\pi  t H_{\r,f}, \quad 0\leq t\leq R,
\een
with $H_{\r,f}$ a positive, essentially selfadjoint operator.
\begin{theorem}\label{main}
In the above setting, let $\f_t$, $\psi_t$ be the restrictions of $\f$, $\psi$ to $\A(W_{f_t})$ and $S(t) \equiv S(\f_t |\!| \psi_t)$.  
We  have
\ben\label{Sf}
S(t)   -  S(0) = -  2\pi t E_{\r,f},\quad 0 \leq t \leq R\ ,
\een
with $ E_{\r,f} = (\xi, H_{\r,f} \xi)$ the mean relative energy associated with $f$. 

If $f$ is constant, $f(y) = a >0$, then $H_{\r, f} = aH_{\r^+}$ with $H_{\r^+}$ the generator of the one-parameter unitary group of null $\mathbf u$-translations in the representation $\r$. 
\end{theorem}
\begin{proof}
The above discussion shows the relation \eqref{Sf} as a consequence of Corollary \ref{deltaS}. If $f$ is constant equal to $a$, then the equality follows because
\[
H_{\r,a} = \frac{1}{2\pi}
(\log \Delta_{\eta, \xi, \M_a} - \log \Delta_{\eta, \xi, \M}) =   K_{\r} - K_{\r,a}  = a(K_{\r} - K_{\r,1}) = aH_{\r^+} ,
\]
where $K_\r$, resp.  $K_{\r,a}$, is the generator of the boost one-parameter unitary group preserving $W$, resp. $W + a {\mathbf u}$. 
\end{proof}
Let $0\leq R < \tilde R$. 
We now consider the case of a charge $\r$ localised in $W_{\f_R}\cap W'_{f_{\tilde R}}$, thus $\r$ supported on the null horizon of $W$ (cf. \cite{L00}). As above, $\r$ may be the composition of several charged $\r = \r_1\r_2\cdots\r_n$. By Theorem \ref{main} we have:
\[
S'(t)   =  -  2\pi E_{\r,f},\quad  t \leq R\ ,\qquad S'(t)  = 0,\quad t\geq \tilde R \ .
\]
Now, the statement that $H_{\r,f}$ is a positive operator is a form of the ANEC. Taking the expectation value on the vacuum state gives the following. 
\begin{corollary}We have:
\[
E_{\r,f} = \frac{1}{2\pi} \int_{-\infty}^{+\infty} S''(t){\rm d}t  > 0 \ .
\]
\end{corollary}
\begin{proof}
Analogous to \eqref{Av}, with $E_{\r,f} > 0$ because $H_{\r,f}$ is a positive operator. 
\end{proof}
In particular, we may take $\r = {\rm Ad}U$ with $U$ a unitary in $\A(W_f)$ (cf. Remark \ref{rem}). Then $H_{\r,f} = U H_f U^*$ and we have
\[
(\eta, H_{f}\eta) = \frac{1}{2\pi} \int_{-\infty}^{+\infty}\frac{{\rm d}^2}{{\rm d}t^2} S(\psi_t |\!|\f_t){\rm d}t  > 0 \ .
\]
A physical argument in \cite[(A.4)]{LLS} in terms of the stress-energy tensor $T_{\mu\nu}$ gives
\ben\label{Hf}
\int_{-\infty}^{+\infty} \frac{{\rm d}^2}{{\rm d}t^2} S(\psi_t |\!|\f_t){\rm d}t  = 2\pi\int f(y)(\eta, T_{uu}\eta){\rm d}^{n-1}y {\rm d}u  \ ,
\een
(null coordinates $u = \frac{1}{\sqrt{2}}(x_0 + x_1)$, $v = \frac{1}{\sqrt{2}}(x_0 - x_1)$, $y_k = x_k$, $k\geq 2$),
showing that
\ben\label{HT}
(\eta, H_{f}\eta) = \int f(y)(\eta, T_{uu}\eta){\rm d}^{n-1}y {\rm d}u  \ .
\een
In particular, we may take $\r = {\rm Ad}U$ with $U$ a unitary in $\A(W_f)$ (cf. Remark \ref{rem}). Then formula \eqref{HT} holds
with $\eta = U\xi$. The linear span of vectors of this form is a dense set $\cal D$ in the Hilbert space by the Reeh-Schlieder theorem, moreover the intersection of $\cal D$ with the domain of $H_f$ can be shown to be still dense (by an averaging procedure)
indeed a core for $H_f$ because $T_f(t)\cal D \subset \cal D$, $t\geq 0$;
we thus infer that
\[
H_f = \int f(y) T_{uu}{\rm d}^{n-1}y {\rm d}u \ ,
\] 
namely $ \int f(y) T_{uu}{\rm d}^{n-1}y {\rm d}u$ is a positive operator, that is the ANEC is always satisfied. 

\section{$U(1)$-current model}
In this section, we illustrate our results by explicit computations in the local conformal net  $\A$ on $\mathbb R$ generated by the $U(1)$-current (see \cite{BMT}). We shall give a complete description of the relative entropy function $S(t)$ associated with a charged state. 

If $h,k$ are in the one-particle Hilbert space, say $h,k$ are real functions in the Schwartz
space $S(\mathbb R)$, we have the commutation relations for the Weyl unitaries on the Bose Fock Hilbert space:
\[
W(h)W(k) = e^{i\int k' h} W(k)W(h)\ .
\]
The norm of $h$ is given  in momentum space by $||h||^2 = \int_0^\infty p|\hat h(p)|^2{\rm d} p$, so we may consider the Weyl unitary  $W(h)$ associated also with any real function $h$ with finite norm. 

If $I\subset \mathbb R$ is an interval or half-line, $\A(I)$ denotes the von Neumann algebra generated by the $W(h)$ with supp$(h) \subset I$. 
 
Let $U$ be the unitary representation of $SL(2,\mathbb R)$ on the Fock Hilbert space; then
\[
U(g)W(h)U(g)^* = U(h_g)\ ,
\]
where $h_g(x) = h(gx)$, $x\in\mathbb R$. 

Let now $\ell$ be real function with compact support on $\mathbb R$. Then we have the localised automorphism $\b_\ell$ studied in \cite{BMT}.   
$\b_\ell$ acts on Weyl unitaries by
\ben\label{bW}
\b_\ell\big(W(h)\big) = e^{-i\int \ell(x)h(x){\rm d}x}W(h)
\een
for every real $h\in S(\mathbb R)$, where $S(\mathbb R)$ is the Schwartz function space. In terms of the $U(1)$-current $j$,
\[
[j(x_1),j(x_2)] =  i\delta'(x_1 - x_2)\ ,
\]
$W(h) = e^{-i\int j(x) h(x){\rm d}x}$ and
$\b_\ell$ is associated with the action 
\[
\b_\ell: j(x)\to j(x) + \ell(x) \ .
\]
The vacuum expectation value of the Weyl unitaries is given by
\ben\label{fW}
\f\big(W(h)\big) = e^{-\frac12 ||h||^2}\ .
\een
The sector class of $\b_\ell$ (i.e. the class of $\b_\ell$ modulo inner automorphisms) is determined by the charge $q\equiv\int \ell(x){\rm d}x$. $\beta$ is inner iff the charge $q$ of $\ell$ is zero and in this case $\beta_\ell =\Ad W(L)$ where $L$ is the primitive of $\ell$, namely $L(x)=\int_{-\infty}^x \ell(a){\rm d}a$.

Let $I$ any interval or half-line and take $\ell_1$ with the same charge as $\ell$ and support contained in $I'$. Since
\[
\b_\ell\big(W(h)\big) =  {\rm Ad}W(\ell - \ell_1)\cdot\b_\ell\big(W(h)\big) = {\rm Ad}W(\ell - \ell_1)\big(W(h)\big)\ , \quad {\rm supp}(h)\subset I \ ,
\] 
$\b_\ell$ is normal on $\A(I)$, namely it extends to a normal map, indeed to an automorphism by eq. \eqref{bW}, of $\A(I)$. It follows that $\b_\ell$ defines a DHR automorphism of $\mathfrak A$ localised in any interval that contains supp$(\ell)$, and that $\b_\ell$ restricts to a normal automorphism of $\A(I)$ for every $I$. 

Let $I$ be an interval with supp$(\ell)\subset I$ and $\cal U$ a connected neighbourhood of the identity in $SL(2,\mathbb R)$ such that ${\rm supp}(\ell)\subset g^{-1}I$ for all $g\in \cal U$. 

We set 
\ben\label{gl}
\ell_g(x) =  \frac{{\rm d}(gx)}{{\rm d} x}\ell(gx),\quad g\in \cal U,
\een
and let as above $L$ be the primitive of $\ell$, $L(x) = \int_{-\infty}^x \ell(a){\rm d}a$ and $L_g$ the primitive of $\ell_g$, thus  
\ben\label{gL}
L_g(x) = L(g x)\ .
\een 
As $\ell_g - \ell$ has zero charge, $L - L_g$ has compact support contained in $I$ and the Weyl unitary $W(L - L_g)$ belongs to $\A(I)$.

We now give the formula for the covariance cocycle $u_g$ of $\b_\ell$ (cf. \cite{BMT}). By the cocycle identity, it is enough to   express of $u_g$ for $g$ in a neighbourhood of the identity in $SL(2,\mathbb R)$, see \cite{GL96}. Put $N= \frac12 q^2$ (the spin).
\begin{proposition}\label{covc}
The covariance unitary cocycle of $\b_\ell$ is given by
\[
u_g = W(L - L_g)e^{\frac{i}{2} \int \ell L_g }e^{-iN/2}, \quad g\in\cal U\ .
\]  
\end{proposition}
\begin{proof}
Let's check the cocycle property of $u_g$. Set $z_g \equiv W(L - L_g)$. We have
\begin{align}
z_{gh} &= W(L - L_{gh})\\ 
&= W(L - L_g + L_h - L_{gh}) \\ 
&= W(L - L_g)W(L_g - L_{gh}) e^{-\frac{i}{2}\int (\ell - \ell_s)(L_g - L_{gh})}\\
&= z_g\a_g(z_h) e^{-\frac{i}{2}\int (\ell - \ell_g)(L_g - L_{gh})}
\end{align}
with $\a_g \equiv {\rm Ad}\,U(g)$, $g,h, hg\in\cal U$. So it suffices to show that
\[
\int (\ell - \ell_g)(L_g - L_{gh}) =  \int \ell L_g  + \int \ell L_h  -\int  \ell  L_{gh}  - N \ . 
\]
Indeed
\begin{align*}
\int (\ell - \ell_g)(L_g - L_{gh}) &=  
  \int \ell L_g   +\int \ell_g  L _{gh}  -\int  \ell  L_{gh} -  \int \ell_g L _{g}\\
  &=     \int \ell L_g   +\int \ell  L _{h}  -\int  \ell  L_{gh} -  \int \ell L \ ,
\end{align*}
so the cocycle property holds because
\[
\int  \ell L = \int \Big(\frac{{\rm d}}{{\rm d}x}L\Big)L  = \frac12 \int \frac{{\rm d}}{{\rm d}x}L^2 = \frac12 L^2(+\infty) = \frac12 q^2 = N \ .
\]
\end{proof}
In particular, the covariance unitary cocycle for dilations is given by
\ben\label{dc}
u_{\delta_s} = W(L - L_{\delta_s})e^{\frac{i}{2} \int \ell L_{\delta_s} }e^{-iN/2},\quad s\in \mathbb R \ ,
\een
where $\delta_s : x\mapsto e^s x$.  

Now suppose {\rm supp}$(\ell) \subset [0, +\infty)$. 
As $d(\b_\ell) = 1$, we have 
\[
u_{\delta_{2\pi s}} = (D\psi : D\f)_{-s}\ ; 
\]
here $\f$ and $\psi = \f\cdot\b_\ell^{-1}$ are restricted to $\A(0,\infty)$.
\begin{proposition}\label{eu}
Let $t\in\mathbb R$ be such that {\rm supp}$(\ell) \subset [t, +\infty)$. We have
\[
S(t) \equiv S(\f_t|\!|\psi_t) = \pi\!\int (x-t)\ell^2(x){\rm d}x\ ,
\]
with $\f_t$ and $\psi_t$ the restrictions of $\f$ and $\psi = \f\cdot\b_\ell^{-1}$  to $\A(t,\infty)$.\footnote{As $S(\f\cdot\b_\ell^{-1}|\!|\f) = S(\f|\!|\f\cdot\b_\ell) =  S(\f|\!|\f\cdot\b^{-1}_{-\ell})$, we have the symmetry
$S(\f_t|\!|\psi_t)  = S(\psi_t|\!|\f_t)$ in this case (cf. Thm. \ref{eta}).}
\end{proposition}
\begin{proof}
We first assume that $t=0$ and {\rm supp}$(\ell) \subset [0, +\infty)$. By equations \eqref{fW}, \eqref{dc} we have
\[
\f(u_{\delta_{s}}) = e^{-\frac12 ||L_{\delta_s} - L||^2}e^{\frac{i}{2} \int \ell L_{\delta_s} }e^{-iN/2}\ .
\]
Taking into account that
\[ 
\frac{\rm d}{{\rm d}s}e^{-\frac12 ||L - L_{\delta_{s}}||^2}\big|_{s = 0} =  -e^{-\frac12 ||L - L_{\delta_{s}}||^2}||L - L_{\delta_{s}}||\,\frac{\rm d}{{\rm d}s} ||L - L_{\delta_{s}}||\,\big|_{s = 0} = 0\ ,
\]
we then have
\begin{align}
\frac{\rm d}{{\rm d}s}\f(u_{\delta_{s}})\big|_{s = 0} 
&= \frac{\rm d}{{\rm d}s}e^{-\frac12 ||L - L_{\delta_{s}}||^2} e^{\frac{i}{2} \int \ell L_{\delta_{s}} }e^{-iN/2}\big|_{s = 0}\\
&= e^{-iN/2}\frac{\rm d}{{\rm d}s}e^{\frac{i}{2} \int \ell L_{\delta_{s}} }|_{s = 0}\\
&= \frac{i}{2}e^{-iN/2}e^{\frac{i}{2} \int \ell L_{\delta_{s}} }\frac{\rm d}{{\rm d}s}\!\int \ell L_{\delta_{s}} \big|_{s = 0}\\ 
&= \frac{i}{2}e^{-iN/2}e^{\frac{i}{2} \int \ell L_{\delta_{s}} }
\int \ell \frac{\rm d}{{\rm d}s}L_{\delta_{s}}\big|_{s = 0} \\ \label{line}
&= \frac{i}{2}e^{-iN/2}e^{\frac{i}{2} \int \ell L_{\delta_{s}} }\int \ell(x) \ell_{\delta_{s}}(x) \frac{{\rm d}(\delta_{s}x)}{{\rm d}s}{\rm d}x \big|_{s = 0}\\
&= \frac{i}{2}e^{-iN/2}e^{\frac{i}{2} \int \ell L_{\delta_{s}} }\int \ell(x) \ell_{\delta_{s}}(x) e^s x{\rm d}x \big|_{s = 0}\\
&= \frac{i}{2}\int x\ell^2(x){\rm d}x \ .
\end{align}
Thus
\[
S(\f|\!|\psi) = S(\f|\!|\f\cdot\b_\ell^{-1}) 
= i\frac{\rm d}{{\rm d}s}\f\big((D\psi : D\f)_{s}\big)\big|_{s = 0}
= -i\frac{\rm d}{{\rm d}s}\f(u_{\delta_{2\pi s}})\big|_{s = 0} = \pi\!\!\int x\ell^2(x){\rm d}x \ .
\]
Let now $\delta^{(t)}$ the one parameter dilation group w.r.t. the half-line $(t,+\infty)$, namely $\delta_s^{(t)}: x\mapsto e^s(x-t)-t$. The above computation works to get $\frac{\rm d}{{\rm d}s}\f(u_{\delta^{(t)}_{s}})|_{s = 0}$ by replacing $\delta$ with $\delta^{(t)}$ up to the line \eqref{line}; since
\[
\frac{{\rm d}(\delta^{(t)}_{s}x)}{{\rm d}s} = e^s(x-t) \ ,
\]
we then have 
\[
\frac{\rm d}{{\rm d}s}\f(u_{\delta^{(t)}_{s}})\big|_{s = 0}  = \frac{i}{2}\int (x-t)\ell^2(x){\rm d}x \ ,
\]
thus
\ben\label{stu}
S(t) 
=  \pi\!\!\int (x-t)\ell^2(x){\rm d}x \ .
\een
\end{proof}
One may compute the infimum $S_{\rm min}(q)$ for the possible entropy of a sector with charge $q$ and charge distribution $\ell$ supported in an interval $I\subset (0,\infty)$. If $I = (1/\lambda , \lambda)$, $\lambda > 1$, 
$S_{\rm min}(q) =  \pi q^2/(2\log\lambda)$ is taken when $\ell(x)$ approaches $q/(2x\log\lambda)$ on $I$ and zero out of $I$. 
We may indeed calculate $S_{\rm min}(q)$ in any given interval $\tilde I$, rather than $(0,+\infty)$, when the charge distribution is concentrated in a fixed interval $I\subset \tilde I$. 
\begin{corollary}\label{cmin}
We have
\ben\label{min}
S_{\rm min}(q) = 2\pi N/\log\nu\ ,
\een
with $\nu$ the cross ratio associated with $I\subset \tilde I$, namely $\nu =  \frac{(c - b)(d - a)}{(c - a)(d - b)}$ if $I = (a,b)$ and $\tilde I = (c, d)$. 
\end{corollary}
\begin{proof}
By the above calculation $S_{\rm min}(q) =  \pi q^2/\log\lambda^2 = 2\pi N/\log\nu$ for the case $(1/\l, \l)\subset (0,\infty)$. 
The corollary then follows by M\"obius covariance.
\end{proof}
\begin{remark} The form of the right hand side in eq. \eqref{min} suggests that a model independent version of Corollary \ref{cmin} should hold true. 
\end{remark}
Now, let $K_\ell$ and $H_\ell$ be the selfadjoint generators of the dilation and translation one parameter unitary groups in the $\b_\ell$ representation. 
As a special case of formula \eqref{st1}, we have here
\ben\label{st3}
S(t) = 2\pi (\xi, K_{\b_\ell}\,\xi) -   2\pi t (\xi, H_{\b_\ell}\,\xi)    
\een
if $\b_\ell$ is localised to the right of $t$.  
By comparing equations \eqref{stu} and \eqref{st3}, we immediately get the following expressions for the mean local energy and mean energy:
\begin{corollary}\label{HK}
We have
\[
(\xi, K_{\b_\ell}\,\xi) = \frac{1}{2}\int x\ell^2(x){\rm d}x ,\qquad (\xi, H_{\b_\ell}\,\xi) = \frac{1}{2}\int \ell^2(x){\rm d}x\ ,
\]
if {\rm supp}$(\ell)\subset (0,+\infty)$. 
\end{corollary}
Let $t\in \mathbb R$ be any point (possibly in the support of $\ell$). Given $\e >0$, choose a smooth real function $ \ell_\e$ such that 
$ \ell_\e(x) = \ell(x)$ if $x\geq t -\e$, with supp$( \ell_\e) \subset (t -2\e,\infty)$ and the family of $ \ell_\e$'s  equibounded. Then
\begin{multline*}
S(t) \leq S({t-\e} )
= S(\f|_{{\A(t -\e,\infty)}}|\!| \f\cdot \b_{ \ell_\e}^{-1}|_{{\A( t-\e,\infty)}})\\
\leq  S(\f|_{{\A(t -2\e,\infty)}}|\!|\f\cdot \b_{ \ell_\e}^{-1}|_{{\A(t -2\e,\infty)}})
=  \pi\!\int_{-\infty}^{+\infty}(x-t) {\ell}_\e^2(x){\rm d}x\ ,
\end{multline*}
thus 
\[
S(t) \leq \lim_{\e\to 0^+}\pi\!\int_{-\infty}^{+\infty}  (x-t)\ell_\e^2(x){\rm d}x =  \pi\!\int_t^{+\infty} (x-t)\ell^2(x){\rm d}x \ .
\]
We shall now show that this inequality is actually an equality. 

Let $\ell$ be a real smooth function with compact support in $\mathbb R$  and $\b_\ell$ the associated localised automorphism. Then the restriction of $\b_\ell$ to $\A(0,+\infty)$ is normal and, by \eqref{bW}, it gives an automorphism of the von Neumann algebra $\A(0,+\infty)$. Denote by  $\b_+ = \b_+(\ell)$ the restrictions of $\b_\ell$ to $\A(0,+\infty)$.  
 
As a general von Neumann algebra relation, we have 
\ben\label{wc}
  {\rm Ad}(D\f_+\cdot\b_+^{-1}: D\f_+)_s\cdot  \s^{\f+}_s\cdot \b_+ \cdot \s^{\f_+}_{-s} = \b_+
\een
with $\f_+$ the restriction of the vacuum state $\f$ to $\A(0,\infty)$ and $\s^{\f_+}$ the associated modular group on $\A(0,+\infty)$.
 
Set $\ell^+(x) = \ell(x)$ if $x\geq 0$ and $\ell^+(x) = 0$ iff $x <0$, and $\ell^+_s(x) = e^s\ell^+(e^s x)$. 
\begin{lemma}
With and $w_s = (D\f_+\cdot\b_+^{-1}: D\f_+)_s$, we have
\ben\label{w-s}
w_{- s} = W(L^+ - L^+_{2\pi s})e^{\frac{i}{2} \int_0^{+\infty} \ell L_{2\pi s} }e^{-iN^+/2}, \quad s\in \mathbb R\ .
\een
Here $L^+(x) = \int_0^x \ell^+(a) {\rm d} a$, $L^+_s(x) = L^+(e^s x)$ and $N^+ = q_+^2/2$, with $q_+  \equiv \int_0^{+\infty} \ell(x){\rm d} x$. 
\end{lemma}
\begin{proof}
The cocycle property of the right hand side in \eqref{w-s} follows as in the proof of Prop. \eqref{covc}. 
By eq. \eqref{wc} and the geometrical meaning of the modular group, we have
\ben\label{Cc}
{\rm Ad}w_{ s}\cdot  {\rm Ad}U(\delta_{- 2\p s})
\cdot \b_+ \cdot {\rm Ad}U(\delta_{ 2\p s}) = \b_+
\een
on $\A(0,+\infty)$. The same covariance relation holds with $w_{-s}$ defined by the the right hand side in \eqref{w-s}. 

Now, $w_s$ is determined up to a phase by eq. \eqref{Cc} because $\A(0,\infty)$ is a factor so, in order to check that \eqref{Cc} actually holds, we have to verify that $w_s$ has the right normalisation, namely, by  \eqref{fW}, we must have
\begin{multline*}
{\underset{s\, \longrightarrow\,  - i}{\rm anal.\, cont.\,}} \f(w_s) =
{\underset{s\, \longrightarrow\,  2\pi i}{\rm anal.\, cont.\,}} \f\big((W(L^+ - L^+_{ s})\big)e^{\frac{i}{2} \int_0^{+\infty} \ell L_{ s} }e^{-iN^+/2}\\
= {\underset{s\, \longrightarrow\,  2\pi i}{\rm anal.\, cont.\,}} e^{-\frac12 ||L^+ - L^+_{ s}||^2}e^{\frac{i}{2} \int_0^{+\infty} \ell L_{ s} }e^{-iN^+/2}
 = 1 \ .
\end{multline*}
One may easily see that, by adding a charge to the right of supp$(\ell)$, we can assume that $\int \ell^+ = 0$, so $L^+$ has compact support. Then by Plancharel theorem
\begin{multline*}
{\underset{s\, \longrightarrow\,  2\pi i}{\rm anal.\, cont.\,}}\frac{i}{2}\int_0^{+\infty} \ell(x) L_{ s}(x){\rm d} x =
{\underset{s\, \longrightarrow\,  2\pi i}{\rm anal.\, cont.\,}}\frac{1}{2}\int \frac1p \hat\ell^+(p) \hat\ell^+(e^{-s}p) {\rm d}p \\
= \frac12 \int \frac1p \hat\ell^+(p) \hat\ell^+(p) {\rm d}p = \frac{i}{4}L^+(+\infty)^2 = iq_+^2/4 =
iN^+/2
\end{multline*}
as the Fourier transform $\hat\ell^+$ of $\ell^+$ is an entire function with exponential decay in the upper half-plane (for the lower half-plane extension of $\int_0^{+\infty} \ell(x) L_{ s}(x){\rm d} x$ one makes an integration by parts). 

The verification that
\[
{\underset{s\, \longrightarrow\,  2\pi i}{\rm anal.\, cont.\,}} ||L^+ - L^+_{ s}||^2 = 0
\]
is similar. 
\end{proof}
The above discussion works as well if $\ell$ is not compactly supported but belongs to the Schwartz space $S(\mathbb R)$. We thus state the following theorem in this generality. 
\begin{theorem} Let $\ell$ be a real function in $S(\mathbb R)$ and $t\in \mathbb R$ any point (possibly in the support of $\ell$). We have
\[
S^-(t)  = -\pi\!\int^{t}_{-\infty} (x-t)\ell^2(x){\rm d}x\ , \qquad
S^+(t)  = \pi\!\int_{t}^{+\infty} (x-t)\ell^2(x){\rm d}x\ ,
\]
where $S^\pm(t) = S(\f_t |\!| \psi_t)$ with $\f_t, \psi_t$ the restrictions of $\f$, $\psi= \f\cdot\b_\ell^{-1}$ to $\A(t,\infty)$ and  to $\A(-\infty,t)$. 
\end{theorem}
\begin{proof} We consider the case of $S(t) \equiv S^+(t)$ and take $t=0$.
Now, along the lines of the proof of Prop. \ref{eu}, we have
\ben\label{ui}
S(0) = i \frac{\rm d}{{\rm d}s} \f(w_{2\pi s}) |_{s=0} = \pi\!\int x\ell^+(x)^2{\rm d}x 
= \pi\int_0^{+\infty}x\ell(x)^2{\rm d}x \ .
\een
The general case in the statement then follows as above. 
\end{proof}
So we have
\[
S'(t) = -\pi\!\int_t^{+\infty} \ell^2(x){\rm d}x \leq 0\ ,
\]
\ben\label{S''}
S''(t) = \pi \ell^2(t) \geq 0 
\een
(see Fig. 3). 
Thus $S(t)$ is decreasing and convex on all $\mathbb R$. Moreover $S(t)$ is smooth and linear on every interval where $\ell(t) =0$. 

Now, by Corollary \ref{HK} and the above discussion, the vacuum energy $E(t, t')$ of the charge $\b_\ell$ in an interval $(t, t')$ is given by
\[
E(t, t') = \frac{1}{2}\int_t^{t'}\ell^2(x){\rm d}x \ ,
\]
thus the vacuum energy  density is
\ben\label{Et}
E(t) = \lim_{t'\to t}\frac{E(t,t')}{t' - t} = \frac{1}{2}\ell^2(t) \ ,
\een
in agreement with the Sugawara formula.

We conclude that the equality $E(t) \geq \frac{1}{2\pi} S''(t)$ is actually an equality. Thus we have the following form of the
Quantum Null Energy Condition:
\begin{corollary}
For the considered states, the QNEC holds with the equality
\ben
E(t) = \frac{1}{2\pi} S''(t) \geq 0
\een
and is not saturated in every point of positive energy density. 
\end{corollary}
\begin{proof}
Compare the equations \eqref{S''} and \eqref{Et}. 
\end{proof}
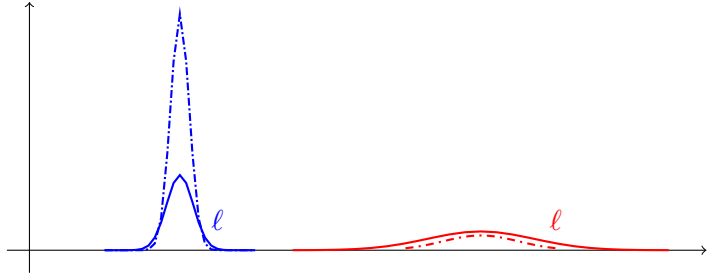
\begin{figure}
\centering
\begin{tikzpicture}
\draw [->] (-0.3,0) -- (9,0); 
\draw [->] (0,-0.3) -- (0,3.3); 
\draw [blue, thick, domain=-1:1] plot (\x + 2 , exp{-\x*\x*16}); 
\draw [blue, thick, densely dashdotted, domain=-1:1] plot (\x + 2 , pi*exp{-\x*\x*32}); 
\draw [red, thick, domain=-2.5 :2.5] plot (\x +6, 1/4 *exp{-\x*\x}); 
\draw [red, thick, dashdotted, domain=-1:1] plot (\x + 6 , 1/16 *pi*exp{-\x*\x*2});
\draw[blue] (2.5,0.4) node { $\ell$};
\draw[red] (7,0.4) node { $\ell$};
\end{tikzpicture}
\caption{\footnotesize Two distributions, blue and red, for the same charge $q$. The dashed lines plot the corresponding entropy density rate $S''(t)$: blue high entropy, red low entropy. }
\end{figure}
\section{Appendix. Commutation relations in the Lie algebra $\mathfrak{sl}(2,\mathbb R)$}
We need to make explicit a couple of commutation relations in the Lie algebra of the group $SL(2,\mathbb R)$.

Let $\g$ be the one parameter subgroup of $SL(2,\mathbb R)$
\ben\label{gs}
\g_s =\left (\begin{matrix} \cosh s/2 & \sinh s/2 \\ \sinh s/2 & 
\cosh s/2\end{matrix}\right ),
\een
and 
\[
\mathfrak a =\frac12\left (\begin{matrix} 0 & 1 \\ 1 & 
0\end{matrix}\right )
\]
the generator of $\g$ in the Lie algebra $\mathfrak{sl}(2,\mathbb R)$ of $SL(2,\mathbb R)$. We have
\ben\label{att'}
\mathfrak a = \frac12(\mathfrak t - \mathfrak t')
\een
with $\mathfrak t , \mathfrak t' \in \mathfrak{sl}(2,\mathbb R)$ given by
\ben\label{tt'}
\mathfrak t = \left (\begin{matrix} 0 & 1 \\ 0 & 
0\end{matrix}\right ),\quad \mathfrak t' = \left (\begin{matrix} 0 & 0 \\ -1 & 
0\end{matrix}\right )\ .
\een
$\mathfrak t$ is the Lie generator of the translation one parameter subgroup $\tau$ of $SL(2,\mathbb R)$
\[
\tau_t = \left (\begin{matrix} 0 & t \\ 0 & 
0\end{matrix}\right )
\]
and $\mathfrak t'$ is conjugate to $\mathfrak t$ by a $\pi$-rotation, i.e. by the adjoint action on $\mathfrak{sl}(2,\mathbb R)$ of the element $\left (\begin{matrix} 0 & 1 \\ -1 & 0\end{matrix}\right )$ of $SL(2,\mathbb R)$ (ray inversion map in the real line picture).
 
If $\delta$ is the dilation one parameter subgroup of $SL(2,\mathbb R)$
\[
\delta_s =\left (\begin{matrix} e^{s/2} & 0 \\ 0  & e^{-s/2} \end{matrix}\right )\ ,
\]
we then have
\ben\label{dil}
\delta_s(\mathfrak a) = \frac12(e^s \mathfrak t - s^{-s} \mathfrak t') \ ,
\een
where $\delta_s(\cdot)$ is the adjoint action of $\delta_s$ on $\mathfrak{sl}(2,\mathbb R)$. 

If $U$ is a unitary representation of $SL(2,\mathbb R)$, we then have the corresponding commutation relations. The same commutation relations hold if $U$ is a unitary representation of the universal cover of $SL(2,\mathbb R)$ because this cover is locally isomorphic to $SL(2,\mathbb R)$. 

\medskip

\noindent
We now  consider the ``$ax + b$'' group, that is the group generated 
two one-parameter unitary groups $\delta$ and $\tau$  (the dilation and the translation groups) that satisfy the commutation relations
\ben\label{eq1}
\delta_s\cdot \tau_t \cdot \delta_{-s} = \tau_{e^s t} \ ,
\een
and is a subgroup thus of $SL(2,\mathbb R)$. Its Lie algebra is generated by $\mathfrak t$ and the generator $\mathfrak d$ of $\delta$. 
We have
\ben\label{cr}
\tau_t(\mathfrak d) = \mathfrak d - t\mathfrak t \ ,
\een
where $\tau_t(\cdot)$ denotes the adjoint action of $\tau_t$. 

Indeed,
by eq. \eqref{eq1} we have
\[
\tau_t \cdot \delta_s\cdot \tau_{-t} = \tau_t \cdot \tau_{-e^s t}\cdot \delta_s = \tau_{(1-e^s)t}\cdot\delta_s \ ,
\]
thus
\[
\tau_t(\mathfrak d) =   \frac{{\rm d}}{{\rm d}s}\tau_t\cdot \delta_s\cdot \tau_{-t} \big\vert_{s=0}
= \frac{{\rm d}}{{\rm d}s}(\tau_{(1-e^s)t}\cdot\delta_s)\big\vert_{s=0} 
=  \mathfrak d- t\mathfrak t \ .
\]
\section{Appendix. Local perturbation of the modular Hamiltonian}
We illustrate here another case that falls within the analysis made in Section  \ref{Wies}. 

Let $\M$ be a von Neumann algebra  on a Hilbert space $\H$ and $\f$ a faithful normal state of $\M$ given by a cyclic and separating vector $\xi\in\H$, $\f = (\xi, \cdot\,\xi)$.
With $\Delta_\xi$ the modular operator of $\xi$, we fix a bounded selfadjoint operator $P\in\M$ and consider the perturbed modular Hamiltonian
\[
\log\Delta_\xi+ P
\]
which is a selfadjoint operator on $\H$. 

There exists a vector $\tilde\eta\in\H$ such that
\[
\log\Delta_{\tilde\eta,\xi} = \log\Delta_\xi + P\ ,
\]
thus 
$
\Delta_{\tilde\eta,\xi} = e^{\log\Delta_\xi + P}
$, and we denote by $\tilde\psi$ the faithful normal positive linear functional on $\M$ associated with $\tilde\eta$.

The Connes Radon-Nikodym unitary cocycle is given by
\ben\label{wsP}
w_s \equiv (D\tilde\psi : D\f)_s= e^{is(\log\Delta_\xi + P)}\Delta_{\xi}^{-is}\ .
\een
We have
\[
\tilde\eta =  w_{-i/2}\, \xi = e^{\frac12(\log\Delta_\xi + P)}\xi\ ,
\]
In particular
\[
\tilde\psi(1) = (\xi, e^{(\log\Delta_\xi + P)}\xi)\ .
\]
The relative entropy between $\f$ and $\tilde\psi$, and between $\f$ and the normalisation $\psi = \tilde\psi/\tilde\psi(1)$ are given by (see \cite{BR})
\[
S(\f|\!|\tilde\psi) = -\f(P)\ ,
\] 
\ben\label{ep}
S(\f|\!|\psi) = -\f(P) +\log(\xi, e^{(\log\Delta_\xi + P)}\xi)\ .
\een
Indeed
\[
S(\f|\!|\tilde\psi) = - (\xi,  \log\Delta_{\tilde\eta,\xi}\xi) = - (\xi, (\log\Delta_\xi + P)\xi)  =  - (\xi,  P\xi) = -\f(P)\ ,
\]
therefore, due to eq. \eqref{scalentr}, 
\[
S(\f|\!|\psi) = S(\f|\!|\tilde\psi) + \log\tilde\psi(1) =   -\f(P) +  \log\tilde\psi(1) = 
 -\f(P) +\log(\xi, e^{(\log\Delta_\xi + P)}\xi)\ .
\]
We assume now that we have a von Neumann subalgebra $\N\subset\M$ which is -hsm w.r.t. $\xi$. 
So we have the translation Hamiltonian $H$ and the tunnel of von Neumann algebras $\M_t$ associated with $\N\subset \M$ and $\xi$. Next proposition shows that we are indeed in the situation considered in Section \ref{Wies}. 
\begin{proposition}
Let $R\geq 0$. 
If $P$ belongs to $\M_R$, then $w_s\in\M_R$ for all $s\leq 0$. We have
\ben\label{prop}
S(\f_t|\!|\psi_t) =  -\f(P) +\log(\xi, e^{(\log\Delta_\xi + P)}\xi)
 - 2\pi t (\xi, H\xi),\quad 0\leq t\leq R\ ,
\een
with $H $ the selfadjoint closure of $ \frac{1}{2\pi  }\big(\log\Delta_{\eta,\xi, \N}- \log\Delta_{\eta,\xi,\M}\big)$. 
\end{proposition}
\begin{proof}
The unitary cocycle $w$ in \eqref{wsP} is the unique solution of the Cauchy problem
\[
-i\frac{\rm d}{{\rm d}s}w_s = w_s\s^\f_s(P),\quad w_0 = 1\ ,
\]
so $w_s \in\M_R$, $s\leq 0$, since $\s^\f_s(P)\in\M_R$ for $s\leq 0$. 

Concerning equations \eqref{prop}, since $\log\Delta_{\xi, \M_t} - \log\Delta_\xi = 2\pi tH$, by \eqref{ep} we have
\[
S(\f|\!|\psi) - S(\f_t|\!|\psi_t)  = (\xi , \log\Delta_{\eta,\xi, \M_t} \xi) -(\xi , \log\Delta_{\eta,\xi}\, \xi) = 2\pi t(\xi, H\xi),\quad 0\leq t\leq R.
\]
\end{proof}

\section{Outlook}
Concerning the equivalence between our form of the ANEC and the one in \cite{BFKLW}, it would be desirable to have a rigorous proof of formula \eqref{Hf}, and of the expression for modular Hamiltonian for $\A(W_f)$ in \cite{CTT}, within the operator algebraic setting or Wightman framework. 
We plan to continue our analysis in a subsequent work. In particular, we intend to describe the relative entropy distribution for the case of a  free field in higher spacetime dimensions. 

\bigskip

\noindent
{\bf Acknowledgements.} 
This paper is the follow up of a question privately set to the author by Edward Witten at the Okinawa Strings 2018 conference. The author warmly thanks him for sharing his insight and constant encouragement. We wish to thank Hirosi Ooguri and the conference organisers for the kind invitation, and Nima Lashkari for comments.

\end{document}